\documentclass[12pt,final,
thmsa,%titlepage,
epsf,a4paper]%{article}
{article}
\usepackage{booktabs}

\newtheorem{theorem}{Theorem}%[chapter]

\newtheorem{corollary}{Corollary}%[chapter]
\newtheorem{definition}{Definition}%[chapter]
\newtheorem{example}{Example}%[chapter]
\newtheorem{lemma}{Lemma}%[chapter]
\newtheorem{proposition}{Proposition}%[chapter]
\newtheorem{remark}{Remark}%[chapter]
\newtheorem{claim}{Claim}%[chapter]

\newenvironment{proof}[1][Proof]{\emph{#1.} }{\  \hfill $\square $ \vspace{5 pt}}

\usepackage{natbib}
\usepackage{amssymb}
\usepackage[dvips]{graphics}
\usepackage{amsmath}

\usepackage{mathpazo}
\usepackage{enumerate}

\usepackage{amsfonts}

\usepackage{amssymb}

\usepackage{enumerate}

\usepackage{mathrsfs}

\usepackage[usenames]{color}
\usepackage{cancel}

\usepackage{pgf,tikz}

\usetikzlibrary{intersections}
\usetikzlibrary{decorations.pathreplacing,angles,quotes}
\usetikzlibrary{arrows.meta}
\usetikzlibrary{arrows, decorations.markings}
\tikzset{myptr/.style={decoration={markings,mark=at position 1 with %
       {\arrow[scale=2,>=stealth]{>}}},postaction={decorate}}}
\usepackage{hyperref}
\hypersetup{
    colorlinks,
    citecolor=blue,
    filecolor=black,
    linkcolor=blue,
    urlcolor=blue
}
\usepackage{nccmath}

\usepackage{pgf,tikz}
\usetikzlibrary{arrows}

\usepackage[utf8]{inputenc}
\usepackage{amssymb, amsmath,geometry}
\usepackage{fancyhdr}
\usepackage{soul}

\usepackage{emptypage}
\newcommand*\samethanks[1][\value{footnote}]{\footnotemark[#1]}

\usepackage[T1]{fontenc}
\DeclareFontFamily{T1}{calligra}{}
\DeclareFontShape{T1}{calligra}{m}{n}{<->s*[1.44]callig15}{}
\DeclareMathAlphabet\mathcalligra   {T1}{calligra} {m} {n}

\begin{document}

\title{Pseudo-Substitutability: A Maximal Domain for Pairwise Stability in Matching Markets with Contracts\thanks{We thank Alejandro Neme and Jordi Massó for valuable discussions during the early stages of this project, and Flip Klijn for helpful comments. We are also grateful to Juan Pablo Torres-Martínez, María Haydée Fonseca-Mairena, and the participants at the University of Chile seminar, supported by project FOVI240104, for their helpful comments and discussions.We acknowledge financial support
from UNSL through grants 031620 and 031323, from Consejo Nacional
de Investigaciones Cient\'{\i}ficas y T\'{e}cnicas (CONICET) through grant
PIP 112-200801-00655, and from Agencia Nacional de Promoción Cient\'ifica y Tecnológica through grant PICT 2017-2355.}}
%\subtitle{Do you have a subtitle?\\ If so, write it here}

%\titlerunning{Short form of title}        % if too long for running head

\author{Nadia Gui\~{n}az\'u\samethanks [2]\thanks{Instituto de Matem\'{a}tica Aplicada San Luis (UNSL-CONICET) and Departamento de Matemática, Universidad Nacional de San
Luis, San Luis, Argentina. Emails: \texttt{ncguinazu@unsl.edu.ar} (N. Gui\~{n}azu), \texttt{nmjuarez@unsl.edu.ar} (N. Juarez), \texttt{pbmanasero@unsl.edu.ar} (P. Manasero), \texttt{paneme@unsl.edu.ar} (P. Neme) and \texttt{joviedo12@gmail.com} (J. Oviedo).}    \and Noelia Juarez\samethanks[2] \and Paola Manasero\samethanks[2] \and Pablo Neme\samethanks[2] \and Jorge Oviedo\samethanks[2]}
%\authorrunning{Short form of author list} % if too long for running head

\date{\today}
\maketitle

\begin{abstract}

We study the existence of pairwise stable allocations in matching markets with contracts and propose a domain restriction that guarantees their existence. Specifically, we define \emph{pseudo-substitutable preferences}, a domain that strictly extends the classical notion of substitutability while still preserving the existence of pairwise stable allocations. This domain accommodates limited complementarities among contracts while retaining enough structure to preserve the key stability properties of substitutable preferences. Moreover, we show that, among all preference domains that contain the classical substitutable domain and guarantee the existence of pairwise stable allocations, the pseudo-substitutable domain is maximal. Our results establish that pairwise stability extends well beyond the classical substitutable domain.

\bigskip

\bigskip


\end{abstract}

\section{Introduction}

In this paper, we study domain restrictions that ensure the existence of pairwise stable allocations in many-to-many matching models with contracts. Our main contribution is to identify a maximal preference domain, among those containing the classical substitutable domain, that guarantees the existence of pairwise stable allocations.

The many-to-many matching with contracts framework provides a unifying model to study markets in which agents on both sides may engage in multiple bilateral relationships. In this setting, the basic unit of analysis is the contract, which specifies both the identity of the parties and the terms of their agreement. This formulation builds on the classical labor-market model of \cite{kelso1982job} and was formalized by \cite{HatfieldMilgrom2005}, who develop the matching-with-contracts approach mainly in a many-to-one environment. Their framework, nonetheless, reveals that the contract-based formulation naturally encompasses a broad class of matching problems, including college admissions, labor markets, and auction environments, and it can be extended to environments with many-to-many interactions. In this paper, we focus on a many-to-many market in which doctors and hospitals may simultaneously engage in multiple contractual relationships. By treating contracts as the objects of choice, the model captures the richness of many-to-many relationships while preserving the tractability needed to analyze the existence and structure of ``stable'' allocations.

Intuitively, an allocation is considered ``stable'' if it cannot be improved upon by agents through the formation of new contractual arrangements. The literature distinguishes two main notions of stability: corewise stability and pairwise stability. In this paper, we focus on the latter notion, namely pairwise stability, which serves as the central stability concept throughout our analysis.

Pairwise stability focuses on bilateral deviations through the formation of a single contract between a doctor and a hospital. An allocation is pairwise stable if every contract outside the allocation is rejected by at least one of the parties involved; equivalently, there is no doctor--hospital pair that can profitably deviate by signing a mutually acceptable contract outside the allocation. Thus, pairwise stability captures the idea that no bilateral agreement provides both agents with an incentive to deviate from the current allocation.

A natural question that arises from this notion is whether such allocations always exist. In general many-to-many matching environments, the answer is negative. For this reason, a central question in the literature is to identify preference domains under which pairwise stable allocations are guaranteed to exist. The classical approach relies on the notion of substitutable preferences, a condition requiring that the desirability of a contract for an agent does not depend on which other contracts are available. In this paper, we revisit this perspective and introduce a weaker domain, which we call \emph{pseudo-substitutable preferences}, that extends the class of markets in which pairwise stability is guaranteed.

The domain of pseudo-substitutable preferences generalizes the classical notion of substitutability while preserving its key structural implications. Preferences in this domain retain many of the desirable features of substitutable preferences—most importantly, the guarantee of pairwise stable allocations—while allowing for certain forms of complementarities among contracts that strict substitutability does not accommodate. 

Intuitively, a complementarity arises when the desirability of a given contract depends positively on the presence of another contract. In other words, an agent may find a contract acceptable only when it is considered together with an additional contract, even though the same contract would not be chosen in isolation. Such complementarities naturally violate substitutability, since the rejection of one contract may alter the desirability of another. The pseudo-substitutable domain is designed precisely to accommodate certain limited forms of complementarity while preserving enough structure to ensure the existence of pairwise stable allocations.

The key idea behind this domain is the notion of a \emph{sub-preference}. Intuitively, a sub-preference is a more restricted version of a preference that preserves its underlying evaluative structure. It is defined on a smaller collection of acceptable sets of contracts, but remains consistent with the way the original preference ranks and evaluates contractual alternatives. More specifically, whenever a set of contracts is acceptable under the sub-preference, it is also acceptable under the original preference. Moreover, the sub-preference preserves the marginal assessment of contracts: if, under the original preference, a contract is desirable when evaluated together with a set of contracts that is acceptable in the sub-preference, then that same contract remains desirable within the sub-preference.
 
A preference is said to be \emph{pseudo-substitutable} whenever it admits at least one sub-preference that is substitutable. In addition, these substitutable sub-preferences are minimal, in the sense that they do not admit any other distinct sub-preference that is also substitutable. This minimality provides a canonical representation of the substitutable structure embedded in the original preference.
 
Our main contribution is to show that, in many-to-many matching markets with contracts, the domain of pseudo-substitutable preferences guarantees the existence of pairwise stable allocations. The proof relies on the fact that every pseudo-substitutable preference admits a substitutable sub-preference. We show that the set of pairwise stable allocations associated with such a substitutable sub-preference is always contained in the set of pairwise stable allocations of the original preference. As a consequence, any method that produces a pairwise stable allocation under substitutable preferences can be applied to the corresponding substitutable sub-preference, and the resulting allocation remains pairwise stable under the full pseudo-substitutable preference profile.
 
Beyond establishing existence, we also show that the pseudo-substitutable domain is \emph{maximal} among all preference domains that contain the classical substitutable domain and guarantee the existence of pairwise stable allocations. More precisely, this means two things: first, whenever all agents’ preferences belong to this domain, at least one pairwise stable allocation exists; second, if the preference of even one agent lies outside the domain, it is possible to find preferences for the remaining agents within the domain such that no pairwise stable allocation exists. Thus, pseudo-substitutability is not only sufficient for existence, but also identifies the largest domain within which pairwise stable allocations can be guaranteed.
 
An additional feature of our framework is that pseudo-substitutability is defined directly at the level of preferences, rather than through the structure of contracts. This makes the domain independent of the contractual representation of the market and allows it to remain meaningful in standard matching models without contracts. In particular, when the analysis is restricted to contract-free many-to-many environments, pseudo-substitutability continues to constitute a strict superdomain of substitutable preferences.

An important implication of this feature is that pseudo-substitutability remains
a strict extension of substitutability even in standard matching models without contracts.
In contrast, most existing relaxations of substitutability collapse to the classical
substitutable domain once the contractual structure is removed.
Thus, pseudo-substitutability captures a genuinely broader class of preferences
that is not driven by contractual richness, but rather by the intrinsic structure
of agents’ preferences.
\subsection*{Pairwise Stability as the Relevant Solution Concept}

The literature distinguishes between corewise and pairwise notions of stability, and these concepts are not equivalent in general.\footnote{For a thorough analysis of the relationship between these two notions of stability in related matching frameworks, we refer the reader to \cite{echenique2004theory}.} 
Corewise stability rules out deviations by arbitrary coalitions of agents through the formation of sets of contracts, whereas pairwise stability restricts attention to bilateral deviations involving a single doctor--hospital contract. As a consequence, corewise stability is generally the stronger requirement.

Under additional structural assumptions, however, the two notions may coincide. In particular, in many-to-one matching markets with substitutable preferences, ruling out bilateral deviations is sufficient to eliminate all profitable coalitional deviations. In this case, pairwise stability and corewise stability are equivalent. This equivalence, however, is no longer guaranteed in many-to-many environments. Even when preferences satisfy substitutability, the possibility that doctors and hospitals may simultaneously engage in multiple contracts gives rise to profitable coalitional deviations that need not be captured by bilateral deviations alone. As a result, pairwise stability and corewise stability generally diverge in many-to-many matching markets.

In light of this distinction, our focus on pairwise stability is motivated primarily by the decentralized nature of the many-to-many markets we study. In many real-world settings, deviations typically arise through bilateral renegotiations rather than through coordinated multilateral agreements. For instance, in decentralized medical labor markets, hospitals commonly negotiate individually with doctors over part-time, temporary, or service-specific contracts. Although coordinated reallocations across multiple hospitals and doctors may in principle improve outcomes, such arrangements are often prevented by institutional, informational, or transactional frictions.

For this reason, pairwise stability provides a more realistic benchmark for the deviations that agents can actually undertake. From an analytical viewpoint, it is also substantially more tractable: verifying pairwise stability only requires checking the absence of blocking doctor--hospital pairs, whereas corewise stability requires evaluating a much larger family of potential blocking coalitions and contract sets. These practical and methodological considerations make pairwise stability the natural solution concept for our analysis.

\subsection*{Related literature}
The literature on stability in matching with contracts investigates which preference domains guarantee the existence of stable allocations in two-sided markets. This question has been studied both in many-to-one markets with contracts---which are inherently \emph{unitary}, since each doctor--hospital pair may sign at most one contract---and in many-to-many markets with contracts, where the distinction between \emph{unitary} and \emph{non-unitary} environments becomes essential. In a unitary many-to-many market, a doctor and a hospital may be involved in several relationships, but they can sign at most one contract with each other. By contrast, in a non-unitary environment, a doctor--hospital pair may simultaneously hold multiple contracts.

 Within this literature, a foundational milestone is \cite{hatfield2008matching}, which revisits the classical claim of \cite{HatfieldMilgrom2005} that substitutability is both necessary and sufficient for the existence of corewise stable allocations in many-to-one markets with contracts. By means of a counterexample, they show that substitutability is not necessary: a hospital may violate substitutability and still admit a corewise stable matching. Their argument relies on a decomposition technique, inspired by \cite{KlausKlijn2005}, which represents a non-substitutable hospital as two virtual entities whose preferences individually satisfy substitutability. Since corewise stability is preserved in the decomposed instance, it also holds in the original market. This result reveals that global substitutability is stronger than required and that stability may survive when preferences admit a suitable substitutable decomposition.
Building on this insight, \cite{hatfield2008matching} introduce the \emph{weak substitutes} condition, which rules out complementarities among contracts involving the same doctor while allowing arbitrary interactions across different doctors. They show that weak substitutes is a necessary condition for corewise stability in many-to-one markets, but not a sufficient one: even under weak substitutes, corewise stable allocations may fail to exist. This leads them to pose an explicit open question: whether there exists a domain that is both necessary (in a maximal sense) and sufficient for corewise stability.

The notion of pseudo-substitutability introduced in this paper addresses this gap, albeit for pairwise rather than corewise stability. While \cite{hatfield2008matching} formulate their question for corewise stability, we show that an analogous structural answer arises when stability is understood in its pairwise form. Pseudo-substitutability achieves this by decomposing preferences---rather than agents---into substitutable sub-preferences, yielding a domain that is both necessary (in a maximal sense) and sufficient for the existence of pairwise stable matchings in many-to-many unitary markets.

Several refinements of substitutability have been proposed in subsequent work. \cite{HatfieldKojima2010} introduce \emph{bilateral} and \emph{unilateral} \emph{substitutability} for many-to-one markets with contracts. Bilateral substitutability is sufficient for corewise stability, whereas unilateral substitutability isolates complementarities within each doctor’s offer set. Although unilateral substitutability alone does not guarantee corewise stability, it plays a central role in the literature on substitutable completability. Pseudo-substitutability is conceptually related to these approaches in that all of them search for minimal relaxations of substitutability that preserve corewise stability. The key difference is both conceptual and methodological. Conceptually, bilateral and unilateral substitutability are tied to corewise stability, whereas pseudo-substitutability is designed to support pairwise stability. Methodologically, bilateral and unilateral substitutability impose behavioral restrictions on choice functions, whereas pseudo-substitutability identifies substitutable sub-components directly within the preference domain.

Along these lines, \cite{Tello2016} studies many-to-one markets with contracts under two-unit demand and irrelevance of rejected contracts (IRC), showing that weak and bilateral substitutability coincide under these assumptions and establishing bilateral substitutability as a maximal domain for corewise stability under bounded complementarities. In contrast, pseudo-substitutability attains maximality for pairwise stability without relying on capacity restrictions.

A parallel strand of the literature investigates \emph{substitutable completability}. \cite{HatfieldKominersHiddenSubstitutes2019} show that in many-to-one markets with contracts, if a hospital’s preferences admit a substitutable many-to-many completion satisfying IRC, then the original market admits corewise stable allocations. \cite{Zhang2016} introduce \emph{weakly observable substitutability} and its \emph{cross-doctor} variant (WOSAD), which require substitutability only along observable offer histories, and analyze their relationship with substitutable completability. \cite{Kadam2017} establishes that unilateral substitutability is sufficient for the existence of a substitutable completion. These contributions explore how substitutable structure can be recovered from broader preference domains. Our approach differs in that it secures pairwise stability without appealing to completions, by extracting substitutable components directly from the original preferences.

The distinction between many-to-many unitary and non-unitary environments is therefore crucial. Pseudo-substitutability applies specifically to many-to-many unitary markets and is thus conceptually related to WOSAD \citep[see][]{BandoHiraiZhang2021}, which—like observable substitutability in \citet{HatfieldKominersWestkamp2021}—operates in the unitary setting. WOSAD captures substitutability failures that arise along observable accumulation paths of cumulative offer mechanisms. In particular, \citet{HatfieldKominersWestkamp2021} show that observable substitutability, together with the associated observable conditions, defines a maximal domain for guaranteeing the existence of corewise stable allocations in cumulative offer mechanisms, in the sense that if these conditions fail, corewise stable allocations may no longer be guaranteed even when other agents have unit-demand preferences. This notion of maximality, however, is mechanism-based and tailored to corewise stability. By contrast, pseudo-substitutability identifies a maximal domain of preferences—defined independently of any mechanism—that guarantees the existence of \emph{pairwise} stable allocations in many-to-many unitary markets with contracts.

A related strand of the literature studies the performance of cumulative offer and deferred acceptance mechanisms beyond standard substitutability assumptions. \cite{KominersSonmez2016} show that corewise stability can be preserved under \emph{slot-specific priorities}, even when substitutability fails, by exploiting an auxiliary agent--slot matching structure. \cite{HatfieldKominersWestkamp2021} provide a full characterization of when corewise stable mechanisms exist in many-to-one matching with contracts, identifying observable conditions under which the cumulative offer mechanism is essentially unique. In a many-to-many setting, \cite{Yenmez2018} study a centralized college admissions clearinghouse and show that deferred acceptance yields corewise stable allocations even when choice rules violate path-independence \citep[and, therefore, substitutability][]{chamb2017choice}, provided they admit path-independent modifications.

In contrast, non-unitary many-to-many markets allow multiple contracts between the same doctor and hospital. \cite{HatfieldKominers2017} show that in these environments substitutability is a maximal domain for corewise stability: if any agent’s preferences violate substitutability, there exist substitutable preferences of the remaining agents such that no stable allocation exists. \citet{BandoHirai2025} extend this perspective through substitutability across doctors and its observable variant, showing that corewise stable allocations may exist even under multi-contract complementarities when doctors’ preferences are responsive.
 While these frameworks share the goal of preserving stability under weaker assumptions, they are not comparable to pseudo-substitutability. The reason is structural: pseudo-substitutability is defined for unitary environments and relies on a decomposition of preferences into substitutable sub-preferences, whereas non-unitary frameworks derive stability from behavioral restrictions on choice correspondences and allow a fundamentally different class of complementarities.

Overall, the literature has developed a broad set of preference domains guaranteeing the existence of stable allocations, predominantly focusing on \emph{corewise} stability in many-to-one and many-to-many markets under various structural, observable, or behavioral restrictions. While substitutability remains a central sufficient condition—and, in some environments, a maximal domain—subsequent work has explored weaker requirements that preserve stability by restricting complementarities or exploiting observable choice behavior. Against this background, pseudo-substitutability offers a complementary perspective: it provides a structural relaxation tailored to many-to-many \emph{unitary} markets with contracts (and thus also to many-to-one markets), delivering a domain that is both necessary (in a maximal sense) and sufficient for the existence of \emph{pairwise} stable allocations. Rather than imposing behavioral constraints or relying on completions, pseudo-substitutability isolates substitutable sub-preferences within the original domain, allowing for non-substitutable interactions outside these components. A comprehensive synthesis of the domains and stability notions discussed in this section can be found in the recent survey by \citet{HatfieldKominers2025}, which provides a unifying treatment of matching with contracts across many-to-one, many-to-many, and networked environments.

The paper is organized as follows. Section~\ref{seccion preliminar} introduces the model and presents the basic notation and preliminary concepts. Section~\ref{seccion de pseudo completa} develops the notion of pseudo-substitutability, states our main results, and examines several structural properties of pseudo-substitutable preferences. Section~\ref{seccion pocicionando a presudo en la lieteratura} discusses the position of the pseudo-substitutable domain relative to other preference domains in the matching literature. Section~\ref{final remark} contains some concluding remarks. Finally, Appendices~\ref{Apendice puebas feas} and  \ref{Apendice puebas de proposiciones} collect a number of technical proofs.

\section{Model and preliminaries}\label{seccion preliminar}

A many-to-many matching model with contracts is specified by a finite set of \emph{doctors} $D$, a finite set of \emph{hospitals} $H$, and a finite set of \emph{contracts} $X$. 
The null contract, denoted by $\emptyset$, represents the situation in which an agent $a \in D \cup H$ signs no contract.

For each contract $x \in X$, let $x_D$ denote the doctor who signs contract $x$, that is, $x_D \in D$, and let $x_H$ denote the hospital that signs contract $x$, that is, $x_H \in H$. 
Given a set of contracts $X' \subseteq X$ and an agent $a \in D \cup H$, let
\[
X'_a=\{x\in X' : a\in \{x_D,x_H\}\}
\]
be the set of contracts in $X'$ involving agent $a$. 
Furthermore, let
\[
X'_D=\{x_D : x\in X'\}
\quad\text{and}\quad
X'_H=\{x_H : x\in X'\}
\]
denote the sets of doctors and hospitals involved in the contracts of $X'$, respectively.

For each hospital $h\in H$, define the set of \emph{feasible} subsets of contracts by
\[
\mathcal{X}_h=
\left\{
X'\subseteq X_h :
|X'_d|\leq 1 \text{ for every } d\in D
\right\}.
\]
Similarly, for each doctor $d\in D$, define
\[
\mathcal{X}_d=
\left\{
X'\subseteq X_d :
|X'_h|\leq 1 \text{ for every } h\in H
\right\}.
\]
These feasibility conditions mean that, for any given agent, a feasible set of contracts cannot include multiple contracts with the same counterpart.
Furthermore, each agent $a\in D\cup H$ may sign multiple contracts and has a strict preference relation $P_a$ defined over the elements of $\mathcal{X}_a$. 
That is, the preference relation $P_a$ is defined only over feasible sets of contracts.
This restriction is known as the \emph{unitarity} condition of the market; see, for instance, \cite{kominers2012correspondence}.\footnote{The unitarity condition can be imposed in different ways. It may be required directly at the level of allocations or, alternatively, imposed at the level of agents' preferences, as we do here. 
For a comprehensive discussion of unitary and non-unitary markets, see \cite{HatfieldKominers2025}.}

The collection of preference relations for all agents \( a \) is referred to as a \emph{preference profile}, denoted by \( P = (P_a)_{a \in D \cup H} \). Given a preference relation \( P_a \) and a feasible set of contracts \( X' \), we define the \emph{choice function} \( C_a^{P} \), which selects the most preferred feasible subset of \( X' \) according to \( P_a \). Throughout the paper, unless otherwise stated, all sets of contracts under consideration are assumed to be feasible.

A \emph{market} is defined by the tuple \( (D,H,X, P) \). Since the sets of doctors, hospitals, and contracts are fixed throughout the paper, we simply refer to a market by its corresponding preference profile \( P \).

One of the central domain restrictions in the matching with contracts literature is the notion of \emph{substitutability}. Introduced by \citet{kelso1982job} and extended to matching with contracts by \citet{HatfieldMilgrom2005}, substitutability imposes restrictions on how an agent’s choice function responds to changes in the set of available contracts. Formally, it requires that the desirability of a contract for an agent does not depend on the presence of other contracts.

We now recall the standard definition.

\begin{definition}\label{defino contratos substitutos}
Given an agent \( a \in D \cup H \), its preference relation \( P_a \), and a set of contracts \( X' \subseteq X\), we say that the contracts in \( X' \) are \textbf{substitutes for agent $\boldsymbol {a }$} if, for any distinct contracts \( x, x' \in X' \),
\[
x' \in C_a^{P}(X') \text{ implies } x' \in C_a^{P}(X' \setminus \{x\}).
\]
Furthermore, we say that an agent’s preference \( P_a \) is \textbf{substitutable} if all contracts are substitutes for agent \( a \).
\end{definition}

In other words, under substitutable preferences, removing an available contract cannot make another previously rejected contract acceptable. We denote by \( \mathcalligra{D} ~\) the domain of substitutable preferences.

Substitutability rules out \emph{complementarities} between contracts, whereby the attractiveness of a contract depends on the simultaneous availability of other contracts. For instance, a hospital may be willing to sign two contracts jointly because of complementarities in skills, while rejecting each contract when offered in isolation. Formally, such behavior violates substitutability, as illustrated by a preference of the form \( P_a : x^1 x^2,  \emptyset \),\footnote{To ease notation, throughout the paper, whenever we list an agent’s preferences over sets of contracts, we omit curly brackets and commas. That is, a set of contracts such as $\{x^1, x^2\}$ will be written simply as $x^1x^2$ when appearing in preference relations.} under which agent \( a \) accepts the pair of contracts but rejects each contract individually.

The following remark presents an equivalent definition of substitute contracts.
\begin{remark}\label{definicion alternatica de contratos subsitutos con inclusiones}
   Given an agent \(a \in D \cup H\), contracts are said to be \textbf{substitutes for \(\boldsymbol{a}\)} if, for every pair of sets of contracts \(X',X'' \subseteq X\) with \(X' \subseteq X''\), it holds that
\[
C^P_a(X'') \cap X' \subseteq C^P_a(X').
\]

\end{remark}
 %\begin{remark}\label{definicion alternatica de contratos subsitutos}
  %  Given a hospital $h \in H$, its preference relation $P_h$  is substitutable if there is no contracts $x,z \in X$  and a set of contracts $Y\subseteq X$ such that $z \notin C^P_h(Y\cup \{z\})$ and $z \in C^P_h(Y \cup\{x,z\}).$
%\end{remark}
Given that the choice function \(C^P_a\) is induced by an agent's preference relation $P_a$, it satisfies consistency.\footnote{Consistency states that for \(X',X''\subseteq X\), if \(C^P_a(X')\subseteq X''\subseteq X'\) then \(C^P_a(X')=C^P_a(X'')\).} In addition, if the preference is substitutable, then \(C^P_a\) satisfies the \emph{path-independence property}:
\begin{equation}\label{path independence}
    C^P_{a}\big( X' \cup X''\big) = C^P_{a}\big( C^P_{a}(X')\cup X''\big)
\end{equation}
for each \(X',X''\subseteq X\).
This property states that the choice over a set of contracts remains unchanged when the set is arbitrarily partitioned, the choice is applied to one segment, and then the choice is reapplied to the selected contracts along with all contracts from the other segment \citep[see][among others]{aizerman1981general,alkan2002class,martinez2008invariance,chamb2017choice}. 

An immediate implication of consistency, which will be useful throughout the analysis, is the following. Since, by definition of the choice function, \(C^P_a(X')\subseteq X'\), whenever \(X'_a\) satisfies
\[
C^P_a(X')\subseteq X'_a \subseteq X',
\]
consistency applied to the sets \(X'\) and \(X'_a\) yields
\[
C^P_a(X'_a)=C^P_a(X').
\]
That is, restricting attention to any intermediate set that still contains all contracts selected from \(X'\) does not alter the choice outcome.

An \emph{allocation} is any set of contracts $Y\subseteq X$ such that $Y_a\in \mathcal{X}_a$ for each $a\in D\cup H$. 
Let $\mathcal{A}$ denote the set of all allocations.

Given an allocation $Y\in\mathcal{A}$, we say that it is \emph{individually rational} if
\[
C_a^{P}(Y)=Y_a
\quad\text{for every } a\in D\cup H.
\]
Note that, when $P_a$ is a substitutable preference relation, $C^P_{a}(Y) =Y_{a}$ implies, by substitutability, that for each contract $x\in Y_a$ we have  $x\in C^P_{a}( \{x\}) .$ 
Given $Y\in \mathcal{A}$, $x\in X \setminus Y$ is a \emph{blocking contract} for $Y$ if $x\in C^P_{a}(Y\cup \{x\})$ for each $a\in \{x_D,x_H\}$. 
   Furthermore, allocation $Y$ is  \emph{(pairwise) stable} if it is individually rational and there is no blocking contract for $Y$. Given the market $P$, denote by $S(P)$ the set of all \emph{pairwise stable allocations for the market } $P$.\footnote{Several papers in the literature consider an alternative notion of stability known as \emph{corewise stability}. An allocation \( Y \subseteq X \) is \emph{corewise stable} if it satisfies two conditions: (i) \(  C_a(Y) = Y_a \) for each $a\in D\cup H$, and (ii) there is no set of contracts $X'\subseteq X$ with \( X' \cap Y=\emptyset \) and for each $a\in X'_{D\cup H}$,  $    X'_a \subseteq C_a(Y \cup X')$ \citep[see][for more details]{HatfieldMilgrom2005}.} Following \cite{klaus2009stable}, the set of pairwise stable allocations in many-to-many matching markets with substitutable preferences is nonempty.

%The classical matching model without contracts can be viewed as a special case of the matching-with-contracts framework by assuming that the set of contracts $X$ contains at most one contract for each pair $(d,h) \in D \times H$. In this paper, we study both models: with and without contracts.

\section{Pseudo-substitutability}\label{seccion de pseudo completa}
In this section, we introduce a new domain of preferences, which we call \emph{pseudo-substitutable}. This domain provides sufficient—and, in a maximal-domain sense, necessary—conditions for the guaranteed non-emptiness of the set of pairwise stable allocations. It strictly contains the domain of substitutable preferences, thereby extending the scope of existing stability results. The definition of pseudo-substitutability is based on a structural condition formulated through the notion of a \emph{sub-preference}, which we introduce next.

\begin{definition}\label{definicion de SubPref}
Given an agent $a\in D\cup H$ and given two preference relations $P_a$ and $P'_a$, we say that $P'_{a}$ is a \textbf{sub-preference}
of $P_{a},$ denoted by $\boldsymbol{P'_{a}\sqsubseteq P_{a}}$, if for each subset of contracts $X',$ we have:
\begin{enumerate}[(i)]
\item $X'_a=C_a^{P'}(X') $ implies that $X'_a=C_a^P( X'),$  and
\item for each $X_a'=C_a^{P'}(X')$, and each $x\in X\setminus X'$ with , $x\in C_a^{P}( X'\cup \{x\}) $ implies that $x\in C_a^{P'}( X'\cup\{x\}) .$%
\end{enumerate}%[chapter]
\end{definition}
Given a preference profile $P$, we say that $P'$ is a \emph{sub-preference profile} of $P,$ $P'\sqsubseteq P,$ if $P'_{a}\sqsubseteq P_{a}$ for each $a\in D\cup H$.
Given a preference relation $P_a$, we denote by $A(P_a)$ the collection of all sets of contracts that satisfy $X'_a = C_a^P(X')$, which we refer to as the \emph{collection of acceptable sets} of the preference relation $P_a$. Conditions (i) and (ii) are directly related to the concept of pairwise stability. The notion of a sub-preference ensures that the set of stable contracts under the sub-preference is contained within the set of stable contracts of the original preference (see Proposition \ref{P'_fsubsetP_f} below). In this framework, Condition (i), when the set of contracts \( X' \) is interpreted as an allocation, is associated with the notion of individual rationality, while Condition (ii) is related to the notion of a blocking contract. Note that, by definition, any preference relation is trivially a sub-preference of itself.

The following example illustrates a preference relation that is a sub-preference of another, as well as one that is not.

\begin{example}\label{Ejemplo 1}
Consider a hospital $h$ and a set of contracts $\{x ,y,z\}$, and let us consider the following preference relations of hospital $h$:
\[
\begin{array}{rl}
P_{h}: & x yz,\, z,\, x y,\, x ,\, y,\, \emptyset,\\
P'_{h}: & x yz,\, x y,\, z,\, x ,\, y,\, \emptyset,\\
P''_{h}: & z,\, x y,\, x ,\, y,\, \emptyset.
\end{array}
\]
It is straightforward to observe that 
\[
 P'_{h} \sqsubseteq P_{h}   \text{ and }  
P''_{h} \sqsubseteq P_{h}.
\] 
However, the preference relation 
\[
\widetilde{P}_{h}: x y,\, z,\, x ,\, y,\, \emptyset
\] 
does not satisfy $\widetilde{P}_{h} \sqsubseteq P_{h}$. To see this, consider 
$C_h^{\widetilde{P}}(\{x ,y\})=\{x ,y\}$. Then,
\[
z \in C_h^{P}(\{x ,y\}\cup\{z\}) = \{x ,y,z\}, 
  \text{ but }  
z \notin C_h^{\widetilde{P}}(\{x ,y\}\cup\{z\}) = \{x ,y\},
\]
which implies that Condition (ii) in Definition~\ref{definicion de SubPref} is violated.\hfill$\Diamond$
\end{example}

Example~\ref{Ejemplo 1} suggests that the sub-preference relation exhibits a hierarchical structure: when one preference is a sub-preference of another, and the latter is itself a sub-preference of a third, the first remains a sub-preference of the third; that is,
\[
P''_{h} \sqsubseteq P'_{h} \sqsubseteq P_{h}.
\]
The lemma below shows that this feature is not merely a property of the example, but holds in general. Establishing this transitivity is key to the structural role played by sub-preferences in our analysis and will be used repeatedly in the results that follow.

\begin{lemma}\label{lema de trasitividad}
Let $a\in D\cup H,$ and let $P''_{a},P'_{a},P_{a}$ be preference relations such that $P''_{a}\sqsubseteq P'_{a},$ and $P'_{a}\sqsubseteq P_{a}.$ Then, $P''_{a}\sqsubseteq P_{a}.$%
\end{lemma}
\begin{proof}Let $a\in D\cup H,$ and $X'\subseteq X$. 
To prove Condition (i), observe that since $P''_{a} \sqsubseteq P'_{a}$, we have that
$X'_a = C_a^{P''}(X')$ implies $X'_a = C_a^{P'}(X').$
In turn, since $P'_{a} \sqsubseteq P_{a}$, this further implies that
$X'_a = C_a^{P}(X').$
 Thus, $X'_a = C_a^{P''}(X')$ implies $X'_a = C_a^{P}(X')$.
To prove Condition (ii), let $x \in X \setminus X'$ with $X'_a=C_a^{P''}(X')$ be such that $x \in C_a^{P}(X' \cup \{x\})$. Now, since $P'_{a} \sqsubseteq P_{a}$, we have that $x \in C_a^{P}(X' \cup \{x\})$ implies  $x \in C_a^{P'}(X' \cup \{x\})$. Moreover, since $P''_{a} \sqsubseteq P'_{a}$, we have that $x \in C_a^{P'}(X' \cup \{x\})$ implies  $x \in C_a^{P''}(X' \cup \{x\})$. Hence, $x \in C_a^{P}(X' \cup \{x\})$ implies $x \in C_a^{P''}(X' \cup \{x\})$ and, therefore, $P''_{a} \sqsubseteq P_{a}$.
\end{proof}

The second key property of the sub-preference relation is that the set of stable allocations under a sub-preference is contained in the set of stable allocations under the original preference. It is important to note that this property establishes only a set inclusion and does not, by itself, guarantee the non-emptiness of the set of pairwise stable allocations under the original preferences. Nevertheless, this inclusion property is central to our main result. The following proposition formalizes this statement.

\begin{proposition}\label{P'_fsubsetP_f}
Let $P$ and $P'$ be two preference profiles. If $P' \sqsubseteq P$, then $S( P') \subseteq S( P) .$%
\end{proposition}
\begin{proof}
Let $P$ and $P'$ be preference profiles such that $P' \sqsubseteq P$. Consider $Y$ a stable allocation under the profile $P'$. We will prove that $Y$ is stable under profile $P.$ 
First, Condition (i) of Definition~\ref{definicion de SubPref} ensures that if an allocation is individually rational under $P'$, then it is also individually rational under $P$. 
Second, assume for the sake of contradiction that $Y$ is not stable under $P$. Then there is a blocking contract $x \in X \setminus Y$ such that $x\in C^P_{a}(Y\cup \{x\})$ for each $a\in \{x_D,x_H\}$. 
Condition (ii) of Definition~\ref{definicion de SubPref} states that if $x \in C_{x_H}^{P}(Y \cup \{x\})$, then $x \in C_{x_H}^{P'}(Y \cup \{x\})$, and $x \in C_{x_D}^{P}(Y \cup \{x\})$, then $x \in C_{x_D}^{P'}(Y \cup \{x\})$.  Thus,  $x$ is a blocking contract for $Y$ under $P'$ as well, contradicting the assumption that $Y$ is stable under $P'$. 
Therefore, $S(P') \subseteq S(P)$.
\end{proof}

The following example illustrates a market in which the set of stable allocations under a sub-preference profile is contained in the set of stable allocations under the original preference profile. In this particular case, the inclusion is strict. However, this need not hold in general. Indeed, in the extreme case in which a preference profile is trivially a sub-preference profile of itself, the two sets of stable allocations coincide.

\begin{example}\label{ejemplo 2}
Consider $D=\{d_1,d_2,d_3\}$ and  $H=\{h\}$ where $h$ is the hospital from Example~\ref{Ejemplo 1}. Let the set of contracts be $X=\{x ,y,z\}$, and recall $P''_h$ and $P_h$ the preference relations of hospital $h$ from Example~\ref{Ejemplo 1}. Furthermore, consider the doctors' preference relations:
\[
P_{d_1}: x ,\emptyset, ~~~  P_{d_2}: y,\emptyset, \text{~~ and ~~}   P_{d_3}: z,\emptyset.
\]
Now, consider the markets $P=(P_D,P_H)$ and $P''=(P_D,P''_H)$. Note that $S(P'')$ consists solely of the allocation $Y=\{z\}$. In contrast, $S(P)$ consists of the allocations $Y=\{z\}$ and $Y'=\{x ,y,z\}$. Therefore, $S(P'') \subsetneq S(P).$ \hfill$\Diamond$
\end{example}

We are now in a position to formally introduce our condition of pseudo-substitutability. This condition allows us to define a domain that extends the classical domain of substitutable preferences, while still ensuring pairwise stability. Formally,

\begin{definition}
\label{Def P-S}Given an agent $a$'s preference relation $P_{a}$, we say that $P_{a}$ is \textbf{pseudo-substitutable}
if there is a substitutable preference relation $P'_{a}$ such that $P'_{a}\sqsubseteq P_a$.
\end{definition}

Example~\ref{Ejemplo 1} provides a direct illustration of this definition. Indeed, the preference relation \(P_h\) is pseudo-substitutable because it admits the sub-preference \(P''_h\), which satisfies \(P''_h \sqsubseteq P_h\) and is substitutable. Hence, although \(P_h\) itself may display complementarities, it still contains a substitutable sub-preference. % and therefore belongs to the pseudo-substitutable domain.

We say that a preference profile $P$ is \emph{pseudo-substitutable}, if $P_{a}$ is a pseudo-substitutable preference relation for each agent $a\in D\cup H.$ Moreover, it follows directly from the definition of pseudo-substitutable preferences and from the definition of sub-preference that a profile of substitutable preferences is also pseudo-substitutable. Formally,

\begin{remark}
  The domain of pseudo-substitutable preferences contains the domain of substitutable preferences.
\end{remark}

The intuition behind the notions of sub-preference relations and pseudo-substitutability is that they provide a way to extract, from a preference relation with complementarities, the portion of the relation that satisfies substitutability. By identifying—and, to some extent, controlling—the sets of contracts that exhibit substitutable behavior within a preference relation that may otherwise display complementarities, we are able to enlarge the domain over which stability can be guaranteed.

Building on this insight, the following result—one of the main contributions of this paper—shows that the newly introduced domain of pseudo-substitutable preference relations, which strictly contains the domain of substitutable preference relations, always admits a non-empty set of stable allocations. Formally,

\begin{theorem}
If $P$ is a pseudo-substitutable preference profile, then $S( P)\neq \emptyset .$%
\end{theorem}
\begin{proof}If $P$ is a pseudo-substitutable preference profile, then there is a substitutable preference profile $P'$ such that $P'\sqsubseteq P.$ Thus, by Proposition \ref{P'_fsubsetP_f}, we have that $S(P') \subseteq S( P) .$ Since $P'$ is a substitutable preference profile, $S(P') \neq \emptyset .$ Therefore, $ S( P) \neq \emptyset.$ 
\end{proof}

The following theorem complements the previous results by showing that pseudo-substitutability is necessary for the existence of pairwise stable allocations. The intuition behind this result is that, given any preference of a hospital that is not pseudo-substitutable, it is possible to construct pseudo-substitutable preference relations for the remaining agents in the market in such a way that the set of stable allocations is empty. Formally,

\begin{theorem}\label{teorema maximal domain}
If $P_{h}$ is not a pseudo-substitutable relation, then there is a market $P$ with at least one other hospital and two doctors, all of whom have pseudo-substitutable preference relations $P_{-h}$, and a set of contracts $X$ containing at least one contract between each doctor-hospital pair such that $S(P)=\emptyset.$\footnote{Where $P_{-h}$ denotes the preference profile of $D \cup H \setminus \{h\}$.}
\end{theorem}
\begin{proof}
    The proof of this theorem is relegated to the Appendix \ref{Apendice puebas feas}.
\end{proof}

Together, the preceding results imply that the domain of pseudo-substitutable preference relations is maximal, as formally stated in the following corollary.

\begin{corollary}
    The domain of pseudo-substitutable preference relations is maximal.
\end{corollary}

\subsection{Properties of Pseudo-substitutable preferences}\label{subseccion de propiedades}
In this subsection, we examine several structural features of the pseudo-substitutable domain. These results show that pseudo-substitutable preferences may behave quite differently from substitutable ones, while still retaining enough order structure to sustain our existence analysis. On the one hand, a pseudo-substitutable preference may admit multiple distinct sub-preferences, and not all of them need to be substitutable. Moreover, since the preference itself is not necessarily substitutable, the induced choice function generally fails to satisfy the \emph{path-independence} property, even though it remains \emph{consistent}. On the other hand, the substitutable sub-preferences of a pseudo-substitutable preference display a strong form of internal coherence: among all sub-preferences, the substitutable ones are minimal with respect to the sub-preference relation, and, furthermore, all substitutable sub-preferences share the same collection of acceptable sets and preserve Blair’s partial order.

To begin with, recall that in Example~\ref{Ejemplo 1} the preference relation $P_h''$ is substitutable and is a sub-preference of $P_h$, which implies that $P_h$ is pseudo-substitutable. However, the preference relation $P'_h$ is not substitutable. As anticipated above, this illustrates that a pseudo-substitutable preference may admit multiple distinct sub-preferences, although at least one of them must be substitutable.

 Next, observe that, for any pseudo-substitutable preference relation \(P_a\), the induced choice function \(C_a^{P}\) satisfies the \emph{consistency} property by construction. However, since \(P_a\) is not necessarily substitutable, the results of \cite{chamb2017choice} imply that the induced choice function does not, in general, satisfy the \emph{path-independence} property. The following example illustrates this point.
\begin{example}
   Consider the pseudo-substitutable preference relation \(P_a : x  y, \emptyset\). Note that \(P'_a = \emptyset\) satisfies \(P'_a \sqsubseteq P_a\) and is substitutable. Observe that
   \[
   C^{P}_a(\{x ,y\}) = \{x ,y\},
   \]
   whereas
   \[
   C^{P}_a\big(C^{P}_a(\{x \}) \cup \{y\}\big)
   = C^{P}_a(\emptyset \cup \{y\})
   = \emptyset.
   \]
   Therefore, the preference relation \(P_a\) fails to satisfy the path-independence property.
\end{example}

We next examine a structural property of pseudo-substitutable preferences related to the organization of their sub-preferences. For a pseudo-substitutable preference, any substitutable sub-preference occupies a minimal position with respect to the sub-preference relation: it has the smallest collection of acceptable sets among all sub-preferences, in the sense of strict inclusion. More formally, for $a \in D \cup H$ and preference relations $P'_a$ and $P_a$ with $P'_a \sqsubseteq P_a$, we say that $P'_a$ is \emph{minimal} if there is no preference relation $P''_a$ such that $P''_a \sqsubseteq P_a$ and $A(P''_a) \subsetneq A(P'_a)$.\footnote{Recall that $A(P_a')$ denotes the collection of acceptable sets of the preference $P'_a$.}

\begin{proposition}\label{proposicion de propiedad 1}Let $a \in D \cup H$, and let $P'_a$ and $P_a$ be preference relations such that $P'_{a}\sqsubseteq P_{a}$. 
If $P'_{a}$ is substitutable, then  $P'_{a}$ is minimal.
    \end{proposition}
\begin{proof}The proof of this proposition  is relegated to the Appendix \ref{Apendice puebas de proposiciones}.\end{proof}

Given an agent \( a \), her preference relation \( P_a \), and the induced choice function \( C^{P}_a \), we define a partial order over subsets of contracts following \citet{blair1988lattice}.  
For two subsets of contracts \( X' \) and \( X'' \), 
we say that \( X' \) is \emph{Blair-preferred to} \( X'' \) for agent \( a \) if
\[
X'_a = C_a^{P}(X' \cup X'').
\]

An important feature of our preference domain is that every substitutable sub-preference of a pseudo-substitutable preference shares the same collection of acceptable sets and, moreover, preserves Blair's order. Formally, 

\begin{proposition}\label{proposicion de propiedad 2}
  Let $a \in D \cup H$, and let $P_a$ a pseudo-substitutable preference relation, and let $P'_a$ and $P''_a$ be substitutable preference relations such that $P'_{a}\sqsubseteq P_{a}$ and $P''_{a}\sqsubseteq P_{a}$. Then,  
    \begin{enumerate}
        \item $A(P'_a)=A(P''_a).$
        \item For each $X_1,X_2\in A(P'_a)$, if $X_1=C_a^{P'}(X_1\cup X_2)$, then $X_1=C_a^{P''}(X_1\cup X_2).$
    \end{enumerate}
       \end{proposition}
\begin{proof}
  The proof of this proposition  is relegated to the Appendix \ref{Apendice puebas de proposiciones}.
\end{proof}

The following example illustrates that the substitutable sub-preference relations of a pseudo-substitutable preference relation share the same collection of acceptable sets and preserve Blair's partial order among them.

\begin{example}
   Consider the pseudo-substitutable preference relation 
$$P_a:x yw,\,x yz,\,z,\,x y,\,x w,\,yw,\,x ,\,y,\,w,\,\emptyset.$$ 
Note that 
$$P'_a:z,\,x yw,\,\boldsymbol{x y,\,x w,\,yw},\,x ,\,y,\,w,\,\emptyset$$  
and 
$$P''_a:z,\,x yw,\,\boldsymbol{yw,\,x w,\,x y},\,x ,\,y,\,w,\,\emptyset$$  
are both substitutable sub-preference relations of $P_a$. 
Furthermore, we can observe that in both $P'_a$ and $P''_a$ the Blair partial order among the sets of contracts is preserved. Moreover, the sets highlighted in boldface, which are \emph{incomparable} under Blair's partial order, are precisely the ones whose relative ranking differs across the two sub-preference relations.\hfill$\Diamond$
\end{example}

\section{Positioning Pseudo-Substitutability within the Literature}\label{seccion pocicionando a presudo en la lieteratura}

The papers reviewed in this section study domain restrictions that guarantee the existence of \emph{corewise stable} allocations in matching-with-contracts environments. 
By contrast, our paper focuses on the more tractable notion of \emph{pairwise stability}, which rules out only bilateral deviations. 
This distinction is not merely technical: it reflects a different modeling objective. 
Corewise stability captures collective deviations requiring coordination among multiple agents, whereas pairwise stability describes decentralized markets in which deviations typically occur through individual or bilateral agreements. 
Our analysis therefore complements such literature by identifying the largest preference domain that guarantees pairwise stability, rather than corewise stability.
To properly understand the scope and contribution of our results, it is necessary to review and compare them with the related literature. 
The following subsections carry out this discussion in detail.

\subsection{Relationship to the Weak Substitutability Domain} 
\cite{hatfield2008matching} revisit the claim by \cite{HatfieldMilgrom2005} that substitutability of contracts is not only sufficient but also necessary for the existence of a core-wise stable allocation. They show that this claim does not hold in general by providing a counterexample in which a hospital’s preferences violate substitutability, yet a corewise stable allocation exists. Their argument relies on decomposing the hospital into two virtual sub-hospitals, each satisfying substitutability, and then proving that the stability of the associated decomposed problem implies the corewise stability of the original one. This idea, originally inspired by \cite{KlausKlijn2005}, reveals that corewise stability can survive even when global substitutability fails, provided that the domain can be factored into substitutable components. 

Building on this insight, \cite{hatfield2008matching}  propose the \emph{weak substitutes} condition, a relaxation of substitutability that excludes comparisons between sets containing multiple contracts with the same doctor. They show that while this weaker condition is necessary for the existence of a corewise stable allocation, it is not sufficient, since even pairwise stability may fail under it. The authors close by noting that ``a condition on preferences that is both sufficient and necessary for guaranteeing existence is still an open question.'' 

Conceptually, our notion of \emph{pseudo-substitutability} can be viewed as a response to this open problem. It extends the insight of \cite{hatfield2008matching}  beyond the agent-based decomposition of preferences, identifying instead substitutable \emph{sub-preferences} at the domain level. In doing so, it preserves pairwise stability through a structural relaxation that is both sufficient and necessary for the existence of a pairwise stable allocation, thus providing a constructive counterpart to the question left open in \cite{hatfield2008matching}.

\subsection{Relationship to the Bilateral Substitutability Domain}

\cite{HatfieldKojima2010} study a many-to-one matching-with-contracts model. They introduce a weaker notion, called \emph{bilateral substitutability}, to isolate the components of substitutability that are essential for ``stability''. They show that bilateral substitutability is sufficient to ensure the existence of corewise stable matchings, although it is not a necessary condition.

As noted by \cite{HatfieldKojima2010}, the \emph{bilateral substitutability} condition strictly weakens \emph{weak substitutability}. Hence, the weakly substitutable domain is contained within the bilaterally substitutable domain, although the two notions differ in the types of complementarities they permit.

Our framework pursues a similar objective: identifying minimal relaxations of substitutability that remain compatible with ``stability''. While bilateral substitutability allows limited complementarities across doctors, pseudo-substitutability permits complementarities within preference structures, provided these can be decomposed into substitutable sub-preferences. Both notions preserve ``stability'' through structured relaxations, yet they operate at different analytical levels---behavioral in the case of bilateral substitutability, and structural in the case of pseudo-substitutability. 

For completeness, we introduce the definition of bilateral substitutability.

\begin{definition}
  A preference  $P_{h}$ is \textbf{bilateral substitutable} if there are no contracts $
x,z\in X,$ and $Y\subseteq X,$ such that $ x_D , z_D
\notin  Y_D ,$ $z\notin C_{h}(Y\cup \{z\})$ and $z\in
C_{h}(Y\cup \{x,z\}).$
\end{definition}

Now, we show that the bilateral substitutable domain and the pseudo-substitutable domain are independent, in the sense that neither implies the other.

Let $h$ be a hospital and consider the bilaterally substitutable preference
\[
P_{h}:xz,x, z', z,\emptyset
\]
where the doctor component satisfies \(x_{D} \neq z_{D} = z'_{D}\).
Note that none of the possible sub-preferences of \(P_{h}\) is substitutable.

Now, let $h$ be a hospital and consider the pseudo-substitutable preference 
\[
P_{h}: xz, \emptyset,
\]
where \(x_{D} \neq z_{D}\). 
Observe that the sub-preference 
\(P'_{h} = \emptyset\) is substitutable. 
On the other hand, this preference is not bilaterally substitutable because, 
for $Y =\{ \emptyset\}$, \(P_{h}\) does not satisfy the definition of bilateral substitutability.

This suggests that the relationship among the domains of substitutable, 
bilaterally substitutable, and pseudo-substitutable preferences 
is the one illustrated in Figure~\ref{figura relacion bilateral y pseudo}.

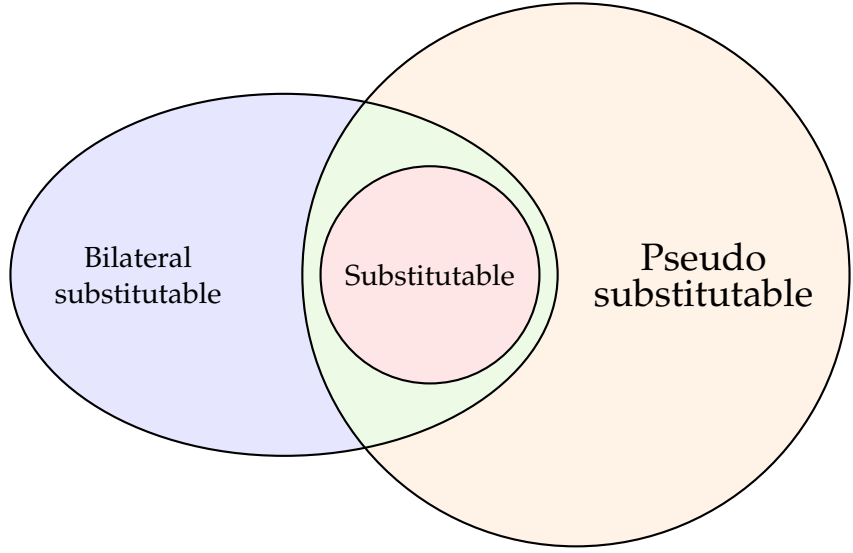
\begin{figure}
    \centering
  \begin{tikzpicture}[scale=1.2]

  % Parámetros
  \def\R{3.0}      % radio del círculo derecho
  \def\d{3.2}      % distancia entre centros

  % Centros
  \coordinate (L) at (-\d/2, 0);
  \coordinate (R) at ( \d/2, 0);

  % --- Conjunto izquierdo: Substitutable completion (ELIPSE) ---
  \fill[blue!10] (L) ellipse (3 and 2);

  % --- Conjunto derecho: Pseudo-substitutable (círculo) ---
  \fill[orange!10] (R) circle (\R);

  % --- Intersección coloreada ---
  \begin{scope}
    \clip (L) ellipse (3 and 2);
    \fill[green!10, opacity=0.7] (R) circle (\R);
  \end{scope}

  % --- Contorno del círculo derecho ---
  \draw[thick] (R) circle (\R);
 \draw[thick]   (L) ellipse (3 and 2);
  % --- Etiquetas ---
  \node at (-3.2,0.2) {\small Bilateral};
  \node at (-3.2,-0.2) {\small substitutable};
  \node at ( 3.0,0.2) {\large Pseudo};
  \node at ( 3.0,-0.2) {\large substitutable};

  % --- Círculo interior: Substitutable ---
  \def\r{1.2}
  \coordinate (I) at (0, 0);
  \fill[red!10] (I) circle (\r);
  \draw[thick]  (I) circle (\r);
  \node at (I) {\small Substitutable};

\end{tikzpicture}
    \caption{Relationship between Bilateral Substitutability and Pseudo-Substitutability}

    \label{figura relacion bilateral y pseudo}
\end{figure}

\subsection{Relationship to the Substitutable Completion Domain}

\cite{HatfieldKominersHiddenSubstitutes2019} show that some preferences that seem to involve complementarities can actually be explained by an underlying substitutable structure. They argue that a non-substitutable many-to-one preference can be extended to a \emph{substitutable completion} in a many-to-many setting, where stability is recovered and then projected back to the original market.

In contrast, while Hatfield and Kominers focus on many-to-one settings and show that substitutable completability guarantees corewise stability, our paper establishes the existence of pairwise stable matchings directly in many-to-many markets. In our framework, \emph{pseudo-substitutability} defines a maximal domain for pairwise stability within the original preference space, without extending it to a contract environment.

For completeness, we introduce the definition of bilateral substitutability.

\begin{definition}
A \textbf{completion} of a hospital's preference relation $P_h$ 
is an extended preference relation $\widehat{P}_h$ 
that coincides with $P_h$ on the feasible subsets of $X_h$. 
A preference relation $P_h$ is said to be a \textbf{substitutable completion} 
if its completion $\widehat{P}_h$ is substitutable.
\end{definition}
 Now, we analyze different scenarios.
 
Let $h$ be a hospital, and consider the non-substitutable preference relation
\[
P_{h}:  xy ,z ,y ,x ,\emptyset,
\]
where \(x_{D} = z_{D} \neq y_{D}\).
Note that $P_h$ can be completed to the substitutable preference relation
\[
\overline{P}_{h}:  xz ,xy ,z ,y 
,x ,\emptyset.
\]
Moreover, the substitutable preference relation
\[
P'_{h}:  z ,xy ,y ,x ,\emptyset
\]
is a sub-preference of $P_{h}$. 
Hence, $P_{h}$ is both substitutably completable and pseudo-substitutable, 
but it is not substitutable.

Now, consider the non-substitutable preference relation 
\[
\widetilde{P}_{h}:  xz ,\emptyset,
\]
where \(x_{D} \neq z_{D}\).
Note that although \(\widetilde{P}_{h}\) is pseudo-substitutable 
(since \(\widetilde{P}'_{h}:\emptyset\) is a substitutable sub-preference relation), 
it cannot be completed to a substitutable preference relation. 
This is because the singletons formed by contracts \(x\) or \(z\) 
cannot appear in any completion that preserves substitutability.
Finally, consider the non-pseudo-substitutable preference relation
\[
\widehat{P}_{h}:  xy ,y,z 
,x ,\emptyset,
\]
where \(x_{D} = z_{D} \neq y_{D}\).
Observe that $\widehat{P}_h$ can be completed to the substitutable preference relation
\[
\widehat{P}'_{h}:  xy ,y ,zx 
,z ,x ,\emptyset.
\]

This suggests that the relationship among the domains of substitutable, 
substitutable completion, and pseudo-substitutable preferences 
is the one illustrated in Figure~\ref{figura relacion completion y pseudo}.

\begin{figure}
    \centering
  \begin{tikzpicture}[scale=1.2]

  % Parámetros
  \def\R{3.0}      % radio del círculo derecho
  \def\d{3.2}      % distancia entre centros

  % Centros
  \coordinate (L) at (-\d/2, 0);
  \coordinate (R) at ( \d/2, 0);

  % --- Conjunto izquierdo: Substitutable completion (ELIPSE) ---
  \fill[blue!10] (L) ellipse (3 and 2);

  % --- Conjunto derecho: Pseudo-substitutable (círculo) ---
  \fill[orange!10] (R) circle (\R);

  % --- Intersección coloreada ---
  \begin{scope}
    \clip (L) ellipse (3 and 2);
    \fill[green!10, opacity=0.7] (R) circle (\R);
  \end{scope}

  % --- Contorno del círculo derecho ---
  \draw[thick] (R) circle (\R);
 \draw[thick]   (L) ellipse (3 and 2);
  % --- Etiquetas ---
  \node at (-3.2,0.2) {\small Substitutable};
  \node at (-3.2,-0.2) {\small completion};
  \node at ( 3.0,0.2) {\large Pseudo};
  \node at ( 3.0,-0.2) {\large substitutable};

  % --- Círculo interior: Substitutable ---
  \def\r{1.2}
  \coordinate (I) at (0, 0);
  \fill[red!10] (I) circle (\r);
  \draw[thick]  (I) circle (\r);
  \node at (I) {\small Substitutable};

\end{tikzpicture}
    \caption{Relationship between Substitutable completion and Pseudo-Substitutability}
    \label{figura relacion completion y pseudo}
\end{figure}
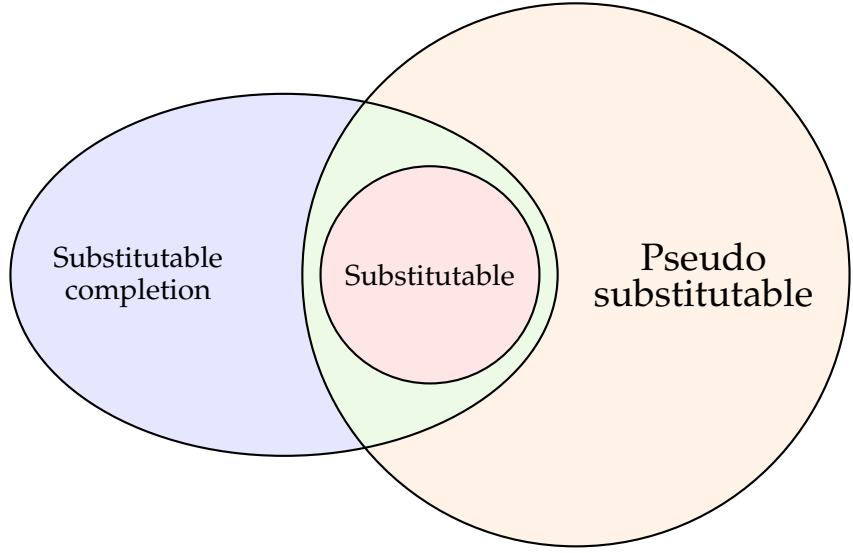

\subsection{Relationship to the WOSAD Domain} 

\cite{BandoHiraiZhang2021} examine how far the substitutability requirement can be relaxed in many-to-many matching with contracts while preserving the existence of corewise stable allocations. The central notion is \emph{weak observable substitutability across doctors} (WOSAD), defined by combining the no substantial violation of substitutability condition with observable irrelevance of rejected contracts, both applied only to offer processes that arise along the cumulative-offer algorithm. Because these restrictions operate exclusively on observable histories, WOSAD is strictly weaker than classical global substitutability requirements.

Compared with \emph{substitutable completability}, WOSAD imposes substantially fewer structural constraints. Under substitutable completability, a hospital's choice function must admit a completion that is substitutable and satisfies irrelevance of rejected contracts. Any such choice function automatically satisfies WOSAD, but the converse fails: WOSAD allows choice functions that cannot be completed to a substitutable one. Thus, the domain of substitutably completable choice functions is properly contained within the WOSAD domain.

The relation with \emph{bilateral substitutability} follows a similar pattern. Bilateral substitutability permits a previously rejected contract to be accepted only when the doctor changes her menu of offers. WOSAD relaxes this restriction further by regulating only those violations that arise along the observable offer paths generated by the cumulative-offer algorithm. Hence, while bilateral substitutability guarantees WOSAD, the reverse does not hold.

WOSAD alone is not sufficient to guarantee the existence of corewise stable allocations in many-to-many markets. The paper provides examples in which hospitals satisfy WOSAD and all other agents have substitutable choice functions, yet no corewise stable allocation exists. The main contribution is a joint sufficiency result: if hospitals satisfy WOSAD and doctors satisfy substitutability together with size monotonicity, then a modified cumulative-offer algorithm always yields a corewise stable allocation. Moreover, WOSAD is not a maximal domain ensuring ``stability'': corewise stable allocations may fail to exist even under stronger assumptions such as bilateral substitutability combined with monotonicity and size monotonicity.

The domain of \emph{pseudo-substitutable preferences} exhibits a nuanced relationship with WOSAD. As discussed above, there are pseudo-substitutable preferences that satisfy substitutable completability and others that satisfy bilateral substitutability. Hence, a substantial subset of pseudo-substitutable preferences is automatically contained in the WOSAD domain. However, this inclusion is not proper. In the example presented below, we display a pseudo-substitutable preference that violates WOSAD, thereby showing that pseudo-substitutability does not imply WOSAD. Taken together, these observations establish that the two domains are \emph{independent}: each contains preference structures that the other does not. Consequently, neither pseudo-substitutability nor WOSAD can be regarded as a refinement of the other, and both represent distinct and non-nested relaxations of substitutability.

Recall the pseudo-substitutable preference $\widetilde{P}_{h}:xz, \emptyset$. Following the approach of \cite{BandoHiraiZhang2021}, 
consider the cumulative offer process $\mathbf{x} = (x,z)$ and its two 
subprocesses $\mathbf{x}^{1} = (x)$ and $\mathbf{x}^{2} = (x,z)$. 
Observe that $x \in \mathbf{x}^{1}$ but $x \notin C_{h}(\mathbf{x}^{1})$. 
Moreover, $x \in C_{h}(\mathbf{x}^{2})$ while 
$x_{D} \notin [C_{h}(\mathbf{x}^{1})]_{D}$. 
Hence, $h$ accepts $x$ in the longer observable subprocess despite the fact 
that $x_{D}$ was not among the doctors chosen in the shorter subprocess. 
This constitutes a substantial observable violation of substitutability. 
Therefore, the cumulative offer process $\mathbf{x}$ fails the WOSAD condition.

Hence, by combining the Venn diagrams in Figures~\ref{figura relacion bilateral y pseudo} 
and~\ref{figura relacion completion y pseudo} with the discussion developed in this subsection, 
we obtain a unified representation of the relationship among substitutability, substitutable completability, 
bilateral substitutability, WOSAD, and pseudo-substitutability. 
This consolidated view is depicted in Figure~\ref{figura relacion todos y pseudo}.

\begin{figure}[h]
    \centering
    \begin{tikzpicture}[scale=1.1]

      % Centros fijos de las elipses / círculos
      \coordinate (L) at (-1.8, 0);   % centro izquierda
      \coordinate (R) at ( 1.8, 1);   % centro derecha
      \coordinate (I) at (0, 0);      % centro del círculo interior
      \coordinate (P) at (0, -1.8);   % centro elipse inferior
      \coordinate (C) at (-1.8, -1.8);% centro círculo WOSAD

      % =========================
      %   RELLENOS (fondo)
      % =========================

      % WOSAD (círculo grande izquierdo/inferior)
      \fill[yellow!15] (C) circle (4.2);

      % Pseudo-substitutable (círculo grande derecho)
      \fill[orange!25, opacity=0.45] (R) circle (4.2);

      % Substitutable completability (elipse izquierda)
      \fill[blue!25, opacity=0.5] (L) ellipse (3.2 and 2.2);

      % Bilateral substitutable (elipse inferior)
      \fill[green!25, opacity=0.45] (P) ellipse (2.2 and 3.2);

      % Substitutable “puro” (círculo interior)
      \fill[red!25, opacity=0.7] (I) circle (1.2);

      % =========================
      %   CONTORNOS (encima)
      % =========================

      \draw[thick] (C) circle (4.2);
      \draw[thick] (P) ellipse (2.2 and 3.2);
      \draw[thick] (L) ellipse (3.2 and 2.2);
      \draw[thick] (R) circle (4.2);
      \draw[thick] (I) circle (1.2);

      % =========================
      %   ETIQUETAS
      % =========================

      % Substitutable completability
      \node at (-3.6,0.4) {\small Substitutable};
      \node at (-3.6,0.0) {\small completion};

      % Bilateral substitutable
      \node at ( 0.0,-3.9) {\small Bilateral};
      \node at ( 0.0,-4.3) {\small Substitutable};

      % Pseudo-substitutable
      \node at ( 3.2,2.4) {\large Pseudo};
      \node at ( 3.2,2.0) {\large substitutable};

      % Substitutable (núcleo interior)
      \node at (I) {\small Substitutable};

      % WOSAD
      \node at (-3.5,-3.5) {\small WOSAD}; 

    \end{tikzpicture}
    \caption{Relationship between the domain of Pseudo-substitutable preferences and other domains.}
    \label{figura relacion todos y pseudo}
\end{figure}
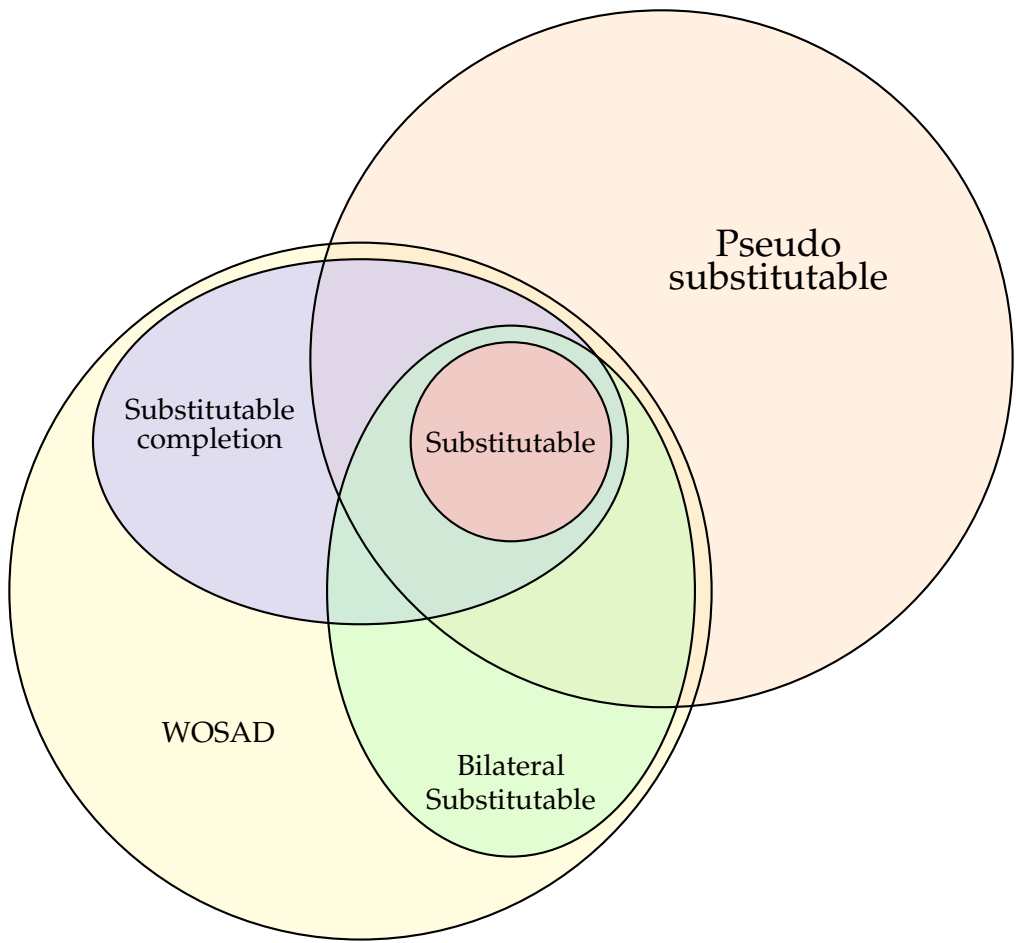

\section{Final Remarks on Our Contribution}\label{final remark}

Our paper introduces the concept of \emph{pseudo-substitutable preferences}, a novel domain that generalizes the classical notion of substitutability. 
Preferences in this domain retain the desirable properties of substitutable preferences—most importantly, the existence of pairwise stable allocations—while allowing for limited complementarities among contracts. 
We formalize this idea through the concept of a \emph{sub-preference}: a simplified version of a preference that preserves its essential structure and remains consistent with the original ordering. 
A preference is pseudo-substitutable if it admits at least one substitutable sub-preference, and all such sub-preferences are minimal and share the same set of acceptable contracts. 
This provides a canonical representation of the substitutable structure within a broader, possibly non-substitutable preference.

Another key distinction concerns the role of contracts.  
All previous papers—from \cite{hatfield2008matching,HatfieldKojima2010} and \cite{HatfieldKominers2017} to \cite{HatfieldKominersHiddenSubstitutes2019,Tello2016} and \cite{BandoHirai2025}—derive substitutability-type conditions defined over \emph{contracts involving the same agent}.  
Their stability results rely critically on the internal structure of these contracts: substitutability, weak substitutability, bilateral substitutability, and related notions all restrict how the acceptance of one contract affects the desirability of others.  
In contrast, our notion of pseudo-substitutability is defined directly over preferences, without imposing any contract-specific restrictions.
Hence, our domain is not tied to the contractual formulation and remains meaningful even in a standard many-to-many matching model without contracts.

This distinction becomes especially clear when comparing contract-based and contract-free settings. 
If one restricts attention to matching models \emph{without} contracts, all the domains proposed in the literature—weak, unilateral, or bilateral, as well as substitutability across doctors—collapse to the standard domain of substitutable preferences. 
That is, once the contractual variables are suppressed, all these refinements coincide with substitutability. 
In contrast, the domain of pseudo-substitutable preferences continues to be a \emph{strict superdomain} of substitutability even in contract-free environments. 
This observation highlights the structural nature of our approach: pseudo-substitutability captures a fundamental relaxation of substitutability that does not rely on the presence of contracts, thereby providing a unified and more general foundation for the analysis of pairwise stability in both contract-based and contract-free matching models.

\subsection*{Some Insights on Stability in Non-Binding Matching Markets}

Many real-world centralized matching markets are inherently non-binding. In such environments, once the clearing house proposes an allocation, agents may unilaterally refuse to accept the assigned contracts. A prominent example is the market for medical residents, where doctors are not forced to accept the matches proposed by the centralized mechanism. In these settings, the classical notion of stability may be insufficient, since an allocation that is stable ex ante may unravel ex post due to individual rejections and, as a consequence, may fail to remain individually rational for some agents.

Our results provide a systematic tool to identify allocations that display a weaker but economically meaningful notion of robustness. Rather than preventing the allocation from unraveling altogether, this notion ensures that, if an allocation breaks down due to the unilateral disengagement of an agent, the remaining allocation continues to satisfy individual rationality. This contrasts with what may occur under corewise stable allocations. In particular, pseudo-substitutability allows us to select pairwise stable allocations that preserve individual rationality even when some agents choose not to accept the proposed allocation.

To illustrate this point, consider a market with a single hospital whose preferences over sets of contracts exhibit complementarities and are given by
\[
P_h:\ zw,xy,x,y,\emptyset,
\]
where all doctors find all contracts acceptable. In this market, the unique corewise stable allocation is $\{z,w\}$, while the set of pairwise stable allocations is $\{x,y,z,w\}$. Although the hospital’s preference is not substitutable, it is pseudo-substitutable, since the sub-preference
\[
P'_h:\ xy,x,y, \emptyset
\]
is substitutable.

Observe that the corewise stable allocation $\{z,w\}$ lacks this robustness in a non-binding environment: if the doctor associated with contract $w$ refuses the match, the hospital is left with contract $z$, which is unacceptable on its own. By contrast, our approach selects the pairwise stable allocation induced by the substitutable sub-preference, namely $\{x,y\}$. In this case, if one of the doctors rejects the proposed match, the hospital remains matched with an acceptable contract.

This example highlights the economic content of pseudo-substitutability. When preferences satisfy this condition, the pairwise stable allocation arising from a substitutable sub-preference guarantees that unilateral disengagements do not lead to violations of individual rationality. In this sense, pseudo-substitutability provides a natural bridge between stability and a notion of robustness that is well suited to non-binding matching markets.

\appendix
\section{Appendix: Proof of Theorem \ref{teorema maximal domain}}\label{Apendice puebas feas}

In this section, we provide the proof of Theorem \ref{teorema maximal domain}. 
We begin by examining how preferences may exhibit complementarities between contracts. 
While substitutable preferences are characterized by the absence of such complementarities, 
in more general preference domains complementary relationships can arise, and these can be classified into distinct types. 
Understanding these distinctions is essential for analyzing the structure of pseudo-substitutable preferences 
and for establishing the maximality result in Theorem \ref{teorema maximal domain}.

Let $h$ be a hospital with preference relation $P_h$. 
Consider a pair of contracts $x^1, x^2 \in X'$, 
where $X' \subseteq X$ is an acceptable set of contracts involving hospital $h$; 
that is, $X' = C^P_h(X')$. 
We first introduce the notion of complementarity. 
We say that $x^1$ is \emph{complementary} to $x^2$ for hospital $h$ 
 if 
$x^1 \notin C^P_h(X' \setminus \{x^2\})$ 
while 
$x^2 \in C^P_h(X' \setminus \{x^1\})$. 
In other words, $x^1$ is only chosen in the presence of $x^2$, 
while $x^2$ can still be chosen without $x^1$. 
We denote this by $x^2 \to x^1$.

If instead both contracts depend on one another—namely, if 
$x^2 \notin C^P_h(X' \setminus \{x^1\})$ 
and 
$x^1 \notin C^P_h(X' \setminus \{x^2\})$—we say that 
$x^1$ and $x^2$ are \emph{bi-complementary}. 
In this case, neither contract is chosen without the other, 
and we write $x^1 \leftrightarrows x^2$. 

A single preference relation $P_h$ over a set of contracts $X'$ 
may exhibit both types of complementarities simultaneously. 
By contrast, under substitutable preferences no complementarities 
between contracts are present.

Before turning to the proof of Theorem \ref{teorema maximal domain}, 
we present a technical lemma that pinpoints how minimal sub-preferences separate 
the pseudo-substitutable domain from its complement. 
If a preference relation is \emph{not} pseudo-substitutable, 
\emph{any} minimal sub-preference necessarily contains a  complementary pair of contracts. 
By contrast, when preferences are pseudo-substitutable, there \emph{exists} a sub-preference 
that is substitutable, and this sub-preference is \emph{minimal}. 
Hence, in the pseudo-substitutable case one can select a minimal sub-preference 
that is free of complementarities, whereas in the non–pseudo-substitutable case 
every minimal sub-preference must exhibit  complementarity.

%\textcolor{red}{We define the following set
%$$
%\mathcal{S}=\{S\subseteq W:\exists~ w^{\star},\widetilde{w}\in C^{P'}_h(S) \text{ with }w^{\star}\notin C^{P'}_h(S\setminus \{\widetilde{w}\}) \text{ and }\widetilde{w}\notin
%C^{P'}_h(S\setminus \{w^{\star}\})\}$$}

%\textcolor{red}{ \begin{lemma}
%Let $P_{h}$ be a no pseudo-substitutable preference and 
%$P_{h}'\sqsubseteq P_{h},$ minimal. Then $\mathcal{S}= \emptyset$
%\end{lemma} }

\begin{lemma}\label{P_h no pseudo, entonces unicomplementaria}
Let $h\in H$, $P_{h}$ be a preference relation that is not pseudo-substitutable. Let  $P'_h$ be any subpreference relation of $P_h$ that is minimal, then there is a set of contracts $X'\in A(P'_h)$ and two contracts $x^{1},x^{2}\in C^{P'}_h(X')$  such that  $x^1$ is complementary to $x^2$, i.e.,
\begin{equation*}\label{ecu eninciado lemma}
    x^{1}\notin
C^{P'}_h(X'\setminus \{x^2\}) ,\text{ and }x^{2}\in
C^{P'}_h(X'\setminus \{x^1\}).
\end{equation*}
\end{lemma}
\begin{proof}
Let $h\in H$, and let $P_h$ be a preference relation that is not pseudo-substitutable. 
Let $P'_h \sqsubseteq P_h$ be minimal. 
To establish the existence of a set $X'\in A(P'_h)$ and two contracts 
$x^{1},x^{2}\in C^{P'}_h(X')$ such that $x^1$ is complementary to $x^2$, 
we first define the following sets:  

$$\mathcal{S}_1=\{X'\in A(P'_h):\exists~ \widetilde{x},x^{\star}\in C^{P'}_h(X') \text{ with }\widetilde{x} \leftrightarrows x^{\star}\}$$
$$\mathcal{S}_2=\{X'\in A(P'_h):\exists~ \widetilde{x},x^{\star}\in C^{P'}_h(X') \text{ with }\widetilde{x}\rightarrow x^{\star}\}$$
Note that if $\mathcal{S}_2\neq\emptyset$, then the result follows. Assume that $\mathcal{S}_2=\emptyset$.  Since $P_h$ is not pseudo-substitutable, $P_h'$ is not a substitutable preference. Hence, $\mathcal{S}_1\neq \emptyset.$ Let $\overline{X}\in \mathcal{S}_1$. Thus, there are $x^1,x^2\in C^{P'}_h(\overline{X})$ such that $x^1\notin C^{P'}_h(\overline{X}\setminus \{x^2\}) $ and $x^2 \notin
C^{P'}_h(\overline{X}\setminus \{x^1\})$.
Since $\overline{X}\in \mathcal{S}_1$, $C^{P'}_h(\overline{X}\setminus \{x^1\})=C^{P'}_h(\overline{X}\setminus \{x^2\})=C^{P'}_h(\overline{X}\setminus \{x^1,x^2\}).$
Next, we define a new preference $P''_h$ and we show that $P''_h$ is a sub-preference of $P_h'$, i.e., $P_h''\sqsubseteq P_h'$, and thus contradicting the minimality of $P_h'.$ Let $P_h''$ defined by:
\begin{enumerate}[(i)]
    \item $C^{P''}_h(\overline{X})=C^{P'}_h(\overline{X}\setminus \{x^1\})=C^{P'}_h(\overline{X}\setminus \{x^2\})=C^{P'}_h(\overline{X}\setminus \{x^1,x^2\}).$
    \item $C^{P''}_h(\widetilde{X})=C^{P'}_h(\widetilde{X})$ for each $\widetilde{X}\neq\overline{X}.$
  \end{enumerate}

   Next, we prove that $P_h''\sqsubseteq P_h'$.  Observe that, by definition of $P_h''$, $A(P_h'')=A(P_h')\setminus C^{P'}_h(\overline{X}),$ implying that $A(P_h'')\subsetneq A(P_h')$. Then, condition (i) of Definition \ref{definicion de SubPref} is fulfilled.
   It remains to prove that for each $\widehat{X}\in A(P'_h)$, and each contract $x$ such that $x_H = h$, the following implication holds:
\begin{equation}\label{ecu 2 lema}
    x \notin \widehat{X} = C^{P''}_h(\widehat{X}) \text{ and } x \in C^{P'}_h(\widehat{X} \cup \{x\}) \text{ imply } x \in C^{P''}_h(\widehat{X} \cup \{x\}).
\end{equation}

Now, there are two cases to analyze, depending on whether $\widehat{X} = \overline{X} \setminus \{x\}$ or not.
 \begin{description}
\item[\textbf{Case 1: }$\boldsymbol{\widehat{X}\neq\overline{X}\setminus\{x\}.}$] Observe that $\widehat{X}\cup\{x\}\neq\overline{X}. $ Assume that \( x \notin \widehat{X} \) and \( x \in C^{P'}_h(\widehat{X} \cup \{x\}) \). By item (ii) in the definition of \( P''_h \), we have \( x \in C^{P''}_h(\widehat{X} \cup \{x\}) \). This proves that  \eqref{ecu 2 lema} holds when $\widehat{X} \neq \overline{X} \setminus \{x\}.$

\item[\textbf{Case 2: }$\boldsymbol{\widehat{X}=\overline{X}\setminus\{x\}.}$]
Assume that $x \notin \widehat{X} = C^{P''}_h(\widehat{X})$ and 
$x \in C^{P'}_h(\widehat{X} \cup \{x\})$. 
We prove that $x \in C^{P''}_h(\widehat{X} \cup \{x\})$.

Suppose otherwise that $x \notin C^{P''}_h(\widehat{X} \cup \{x\})$. 
Since $\widehat{X}=\overline{X}\setminus\{x\}$, it follows that 
$x \notin C^{P''}_h(\overline{X})$. 
By the definition of $P''$, this implies that 
\[
x \notin C^{P'}_h(\overline{X}\setminus\{x^{1}\})
\quad\text{and}\quad
x \notin C^{P'}_h(\overline{X}\setminus\{x^{2}\}).
\]

Since we assume that $\mathcal{S}_2 = \emptyset$, it must also hold that 
\[
x^{1}\notin C^{P'}_h(\overline{X}\setminus\{x\})
\quad\text{and}\quad
x^{2}\notin C^{P'}_h(\overline{X}\setminus\{x\}).
\]

Moreover,
\[
C^{P'}_h(\overline{X}\setminus\{x\})
= C^{P'}_h(\widehat{X})
= \widehat{X}
\]
implies that $x^{1},x^{2}\notin \widehat{X}$, which by the definition of $\widehat{X}$ also implies that 
$x^{1},x^{2}\notin \overline{X}\setminus\{x\}$.

This is a contradiction, since $x^{1},x^{2}\in\overline{X}$ and $x^{1}\neq x^{2}$.

Therefore, $x\in C^{P''}_h(\widehat{X}\cup\{x\})$, which proves that 
\eqref{ecu 2 lema} holds when $\widehat{X} = \overline{X} \setminus \{x\}$.
 
\end{description}
Considering both cases, we conclude that $P_h''$ is a sub-preference of $P_h'$, i.e., $P_h'' \sqsubseteq P_h'$, and by Lemma \ref{lema de trasitividad}, we have that $P_h'' \sqsubseteq P_h$. Thus, contradicting the minimality of $P_h'$. As this contradiction stems from the assumption that $\mathcal{S}_2 = \emptyset$, we deduce that $\mathcal{S}_2 \neq \emptyset$, thereby completing the proof of the lemma.
\end{proof}

The goal of the following proof is to show that, given a hospital whose preference relation is not pseudo-substitutable, we can construct a market with another hospital with a pseudo-substitutable preference relation, together with a set of doctors whose preferences are also pseudo-substitutable, such that the set of stable allocations is empty. 

As will become clear, the preference relation of the second hospital is linear---a special case of pseudo-substitutable (and therefore pseudo-substitutable) preferences---and the doctors’ preferences are linear as well, again representing a special case of pseudo-substitutable preferences. 

This construction defines a many-to-one market, which can be viewed as a special case of a many-to-many market.
\medskip

\noindent \begin{proof}[Proof of Theorem \ref{teorema maximal domain}]
    Since $P_h$ si not pseudo-substitutable preferences then, by Lemma \ref{P_h no pseudo, entonces unicomplementaria}, there is $P'_h\sqsubseteq P_h$ minimal that is not a substitutable sub-preference. Then, there is a set of contracts $X'\subseteq X$ and two contracts involving hospital $h$ such that $x^1,x^2\in X'=C^{P'}_h(X')$ such that $x^1\to x^2$. That is, 
    \begin{equation}\label{ecu1 teorema maximal domain}
        x^1\in C^{P'}_h(X'\setminus \{x^2\}) \text{ and } x^2\notin C^{P'}_h(X'\setminus \{x^1\}).
    \end{equation}
W.l.o.g. assume that $X'$ is minimal in the sense that there is no $X''\subsetneq X'$ fulfilling Condition \eqref{ecu1 teorema maximal domain}.

There are three cases to consider:
\begin{description}
    \item[\textbf{Case 1: There is only one pair of contracts fulfilling Condition \eqref{ecu1 teorema maximal domain}}. ] 
        Define $\overline{X}=X'\setminus \{x^1,x^2\}$. Then, Condition \eqref{ecu1 teorema maximal domain} is equivalent to  
 \begin{equation}\label{ecu2 teorema maximal domain}
        x^1\in C^{P'}_h(\overline{X} \cup \{x^1\}) \text{ and } x^2\notin C^{P'}_h(\overline{X} \cup \{x^2\}).
    \end{equation}
Now, consider the following preference $P'_h$ satisfying Condition \eqref{ecu2 teorema maximal domain}:

$$P'_h: \ldots, \overline{X}\cup \{x^1,x^2\}, \overline{X}\cup \{x^1\},\overline{X},\ldots$$

Define also the following preferences for the rest of the agents:
\begin{align*}
    &P'_{h'}:y^2,y^1,\emptyset&\\
    &P'_{d^1}:y^1,x^1,\emptyset&\\
    &P'_{d^2}:x^2,y^2,\emptyset&\\
    &P'_{d}:z,\emptyset&\text{ for each } d\in \overline{X}_D \text{ and } z\in \overline{X}\\
    &P'_{d}:\emptyset&\text{ for each } d\in D \setminus (\overline{X}_D\cup \{d^1,d^2\})
\end{align*}

Note that, in order to prove that there is no stable allocation under this preference profile, we first prove the following claim:

\begin{claim}\label{Claim del teorema}
For each $\widetilde{X}\subsetneq \overline{X}$ and each $z\in \overline{X}\setminus\widetilde{X} $, we have that $z\in C^{P'}_h(\widetilde{X}\cup \{z\}).$    
\end{claim}
Assume not, i.e., $z\notin C^{P'}_h(\widetilde{X}\cup \{z\}).$ By the assumption of the minimality of $\overline{X}$, there is  a contract $z'\in\overline{X}\setminus \widetilde{X}$ such that $z'\notin C^{P'}_h(\widetilde{X}\cup \{z'\}).$ Then, following the same arguments as in the proof of Lemma \ref{P_h no pseudo, entonces unicomplementaria}, it is possible to construct a sub-preference of $P'$ violating the minimality of $P'.$ Then, $z\in C^{P'}_h(\widetilde{X}\cup \{z\}),$ proving the Claim.

Now, we can analyze the possible individually rational allocations and verify that there is no one stable.
\begin{center}
    \begin{tabular}{|c|c|c|}\hline
 $Y_h$    & $Y_{h'}$ & Blocking contract\\  \hline 
 $\overline{X}\cup \{x^1,x^2\}$& $\emptyset$  &  $y^1$ \\
 $\overline{X}\cup \{x^1\}$& $y^2$  &  $x^2$ \\
  $\overline{X}\cup \{x^1\}$& $\emptyset$  &  $x^2$ \\
 $\overline{X}$&  $y^2$ &  $x^1$ \\
 $\overline{X}$& $y^1$  & $y^2$  \\
 $\overline{X}$& $\emptyset$  & $y^2$  \\
 $\widetilde{X}\cup A$ & $B$  &  $z$ \\\hline
%  $S'\cup \{x^1\}$ &  $x^1$ &  $(h,d)$ \\
%   $S'\cup \{x^2\}$ &  $x^2$ &  $(h,d)$ \\ 
\end{tabular}
\end{center}
 where $\widetilde{X}$ is such that $\widetilde{X}\subsetneq\overline{X}$, $A\in2^{\{x^1,x^2\}}$, $B\in\{\emptyset,y^1,y^2\}$ and, by the Claim, $z\in C^{P'}_h(\widetilde{X}\cup \{z\})$ together with the fact that $P'_{d}:z,\emptyset$ for each  $d\in \overline{X}_D\setminus \widetilde{X}_D$.\footnote{It is important to note that the sets \( A \) and \( B \) must be chosen in such a way that Y remains an allocation.}

\item[\textbf{Case 2: There are two or more pairs of contracts fulfilling Condition \eqref{ecu1 teorema maximal domain}.}] 
      We analyze the case with two pairs of complementary contracts, that is, $x^1\to x^2$ and $x^3\to x^4$. For more than two pairs of contracts, the proof is analogous. Define $\overline{X} = X' \setminus \{x^1, x^2, x^3, x^4\}$. Then, Condition~\eqref{ecu1 teorema maximal domain} applied to these pairs yields:
 
 \begin{equation}\label{ecu3 teorema maximal domain}
 \begin{split}
        x^1\in C^{P'}_h(\overline{X} \cup \{x^1,x^3,x^4\}) \text{ and } x^2\notin C^{P'}_h(\overline{X} \cup \{x^2,x^3,x^4\}),\text{ and}\\
        x^3\in C^{P'}_h(\overline{X} \cup \{x^1,x^2,x^3\}) \text{ and } x^4\notin C^{P'}_h(\overline{X} \cup \{x^1,x^2,x^4\}).
        \end{split}
    \end{equation}
Now, consider the following preference $P'_h$ satisfying Condition \eqref{ecu3 teorema maximal domain}:

$$P'_h: \ldots, \overline{X}\cup \{x^1,x^2,x^3,x^4\},  \overline{X}\cup \{x^1,x^3,x^4\},  \overline{X}\cup \{x^1,x^2,x^3\},\overline{X}\cup \{x^1,x^4\}, \overline{X}\cup \{x^3,x^4\},$$ $$\overline{X}\cup \{x^1,x^2\}, \overline{X}\cup \{x^2,x^3\}, \overline{X}\cup \{x^1,x^3\}, \overline{X}\cup \{x^1\}, \overline{X}\cup \{x^2\}, \overline{X}\cup \{x^3\}, \overline{X}\cup \{x^4\}, \overline{X},\ldots$$

Define also the following preferences for the rest of the agents:
\begin{align*}
    &P'_{h'}:y^2,y^4,y^3,y^1\emptyset&\\
    &P'_{d^1}:y^1,x^1,\emptyset&\\
    &P'_{d^2}:x^2,y^2,\emptyset&\\
    &P'_{d^3}:y^3,x^3,\emptyset&\\
    &P'_{d^4}:x^4,y^4,\emptyset&\\
    &P'_{d}:z,\emptyset&\text{ for each } d\in \overline{X}_D \text{ and } z\in \overline{X}\\
    &P'_{d}:\emptyset&\text{ for each } d\in D \setminus (\overline{X}_D\cup \{d^1,d^2,d^3,d^4\})
\end{align*}

Now, we can analyze the possible individually rational allocations and verify that there is no one stable.
\begin{center}
    \begin{tabular}{|c|c|c|}\hline
 $Y_h$    & $Y_{h'}$ & Blocking contract\\  \hline 
 $\overline{X}\cup \{x^1,x^2,x^3,x^4\}$& $\emptyset$  &  $y^1$ \\
 $\overline{X}\cup \{x^1,x^3,x^4\}$& $y^2$  &  $x^2$ \\
  $\overline{X}\cup \{x^1,x^3,x^4\}$& $\emptyset$  &  $x^2$ \\
 $\overline{X}\cup \{x^1,x^2,x^3\}$& $y^4$  &  $x^4$ \\
  $\overline{X}\cup \{x^1,x^2,x^3\}$& $\emptyset$  &  $x^4$ \\
   $\overline{X}\cup \{x^1,x^3\}$& $y^2$  &  $x^2$ \\
    $\overline{X}\cup \{x^1,x^3\}$& $y^4$  &  $x^2$ \\
     $\overline{X}\cup \{x^1,x^3\}$& $\emptyset$  &  $x^2$ \\
     $\overline{X}\cup \{x^1,x^4\}$& $y^3$  &  $y^2$ \\
      $\overline{X}\cup \{x^1,x^4\}$& $y^2$  &  $x^3$ \\
      $\overline{X}\cup \{x^1,x^4\}$& $\emptyset$  &  $x^3$ \\
       $\overline{X}\cup \{x^3,x^4\}$& $y^1$  &  $y^2$ \\
        $\overline{X}\cup \{x^3,x^4\}$& $y^2$  &  $x^1$ \\
          $\overline{X}\cup \{x^3,x^4\}$& $\emptyset$  &  $x^1$ \\
         $\overline{X}\cup \{x^1,x^2\}$& $y^3$  &  $y^4$ \\
 $\overline{X}\cup \{x^1,x^2\}$& $y^4$  &  $x^3$ \\
  $\overline{X}\cup \{x^1,x^2\}$& $\emptyset$  &  $x^3$ \\
  $\overline{X}\cup \{x^2,x^3\}$& $y^1$  &  $y^4$ \\
   $\overline{X}\cup \{x^2,x^3\}$& $y^4$  &  $x^1$ \\
    $\overline{X}\cup \{x^2,x^3\}$& $\emptyset$  &  $x^1$ \\
%$\overline{X}\cup \{x^2,x^4\}$& $x^1$  &  $(h',x^3)$ \\
 %  $\overline{X}\cup \{x^2,x^4\}$& $x^3$  &  $(h,x^1)$ \\
   $\overline{X}\cup \{x^1\}$& $x^2$  &  $x^3$ \\
   $\overline{X}\cup \{x^1\}$& $x^3$  &  $x^2$ \\
   $\overline{X}\cup \{x^1\}$& $x^4$  &  $x^2$ \\
   $\overline{X}\cup \{x^1\}$& $\emptyset$  &  $x^2$ \\
   $\overline{X}\cup \{x^2\}$& $x^1$  &  $x^3$ \\
   $\overline{X}\cup \{x^2\}$& $x^3$  &  $x^1$ \\
   $\overline{X}\cup \{x^2\}$& $x^4$  &  $x^1$ \\
     $\overline{X}\cup \{x^2\}$& $\emptyset$  &  $x^1$ \\
   $\overline{X}\cup \{x^3\}$& $B$  &  $x^2$ \\
   $\overline{X}\cup \{x^4\}$& $B$  &  $x^2$ \\
            $\overline{X}$& $B$  &  $x^2$ \\
               
 %$\overline{X}$&  $x^2$ &  $(h,x^1)$ \\
 %$\overline{X}$& $x^1$  & $(h',x^2)$  \\
 $\widetilde{X}\cup A$ & $B$  &  $z$ \\
   \hline
\end{tabular}
\end{center}
where $\widetilde{X}$ is such that $\widetilde{X}\subsetneq\overline{X}$, $A\in2^{\{x^1,x^2,x^3,x^4\}}$, $B\in\{\emptyset,y^1,y^2,y^3,y^4\}$ and, by the Claim, $z\in C^{P'}_h(\widetilde{X}\cup \{z\})$ together with the fact that $P'_{d}:z,\emptyset$ for each  $d\in \overline{X}_D\setminus\widetilde{X}_D$.\footnote{It is important to note that the sets \( A \) and \( B \) must be chosen in such a way that \( Y \) remains an allocation.}
\item[\textbf{Case 3: There are two or more overlapping pairs of contracts fulfilling Condition \eqref{ecu1 teorema maximal domain}:}]

     We analyze the case of two overlapping pairs of complementary contracts.  Define $\overline{X} =X' \setminus \{x^1, x^2, x^3\}$.\footnote {For more than two pairs, the proof is analogous.} Now, there are four possible cases to consider.
   \begin{description}
       \item[$\boldsymbol{( x^1\to x^2\text{ and }x^2 \to x^3):}$] 
      The application of Condition~\eqref{ecu1 teorema maximal domain} to these pairs of complementary contracts yields the following:
 
 \begin{equation}\label{ecu4 teorema maximal domain}
 \begin{split}
        x^1\in C^{P'}_h(\overline{X} \cup \{x^1,x^3\}) \text{ and } x^2\notin C^{P'}_h(\overline{X} \cup \{x^2,x^3\}),\text{ and}\\
        x^2\in C^{P'}_h(\overline{X} \cup \{x^1,x^2\}) \text{ and } x^3\notin C^{P'}_h(\overline{X} \cup \{x^1,x^3\}).
        \end{split}
    \end{equation}
Now, consider the following preference $P'_h$ satisfying the Condition \eqref{ecu4 teorema maximal domain}:

$$P'_h: \ldots, \overline{X}\cup \{x^1,x^2,x^3\},  \overline{X}\cup \{x^1,x^3\},  \overline{X}\cup \{x^1,x^2\},\overline{X}\cup \{x^1\}, \overline{X}\cup \{x^2\}, \overline{X}\cup \{x^3\},  \overline{X},\ldots$$

Define also the following preferences for the rest of the agents:
\begin{align*}
    &P'_{h'}:y^3,y^2,y^1,\emptyset&\\
    &P'_{d^1}:y^1,x^1,\emptyset&\\
    &P'_{d^2}:x^2,y^2,\emptyset&\\
    &P'_{d^3}:x^3,y^3,\emptyset&\\
       &P'_{d}:z,\emptyset&\text{ for each } d\in \overline{X}_D \text{ and } z\in \overline{X}\\
    &P'_{d}:\emptyset&\text{ for each } d\in D \setminus (\overline{X}_D\cup \{d^1,d^2,d^3\})
\end{align*}

Now, we can analyze the possible individually rational allocations and verify that there is no one stable.
\begin{center}
    \begin{tabular}{|c|c|c|}\hline
 $Y_h$    & $Y_{h'}$ & Blocking contract\\  \hline 
 $\overline{X}\cup \{x^1,x^2,x^3\}$& $\emptyset$  &  $y^1$ \\
 $\overline{X}\cup \{x^1,x^2\}$& $y^3$  &  $x^3$ \\
  $\overline{X}\cup \{x^1,x^2\}$& $\emptyset$  &  $x^3$ \\
 $\overline{X}\cup \{x^1,x^3\}$& $y^2$  &  $x^2$ \\
  $\overline{X}\cup \{x^1,x^3\}$& $\emptyset$  &  $x^2$ \\
 
    $\overline{X}\cup \{x^1\}$& $y^2$  &  $x^2$ \\
     $\overline{X}\cup \{x^1\}$& $y^3$  &  $x^2$ \\
       $\overline{X}\cup \{x^1\}$& $\emptyset$  &  $y^2$ \\
    
       $\overline{X}\cup \{x^2\}$& $y^1$  &  $y^3$ \\
        $\overline{X}\cup \{x^2\}$& $y^3$  &  $x^1$ \\
          $\overline{X}\cup \{x^2\}$& $\emptyset$  &  $y^3$ \\
         
 $\overline{X}\cup \{x^3\}$& $y^1$  &  $y^2$ \\
  $\overline{X}\cup \{x^3\}$& $y^2$  &  $x^1$ \\
  $\overline{X}\cup \{x^3\}$& $\emptyset$  &  $y^1$ \\
            $\overline{X}$& $B$  &  $x^2$ \\
               
 %$\overline{X}$&  $x^2$ &  $(h,x^1)$ \\
 %$\overline{X}$& $x^1$  & $(h',x^2)$  \\
 $\widetilde{X}\cup A$ & $B$  &  $z$ \\
   \hline
\end{tabular}
\end{center}
where $\widetilde{X}$ is such that $\widetilde{X}\subsetneq\overline{X}$, $A\in2^{\{x^1,x^2,x^3\}}$, $B\in\{\emptyset,y^1,y^2,y^3\}$ and, by the Claim, $z\in C^{P'}_h(\widetilde{X}\cup \{z\})$ together with the fact that $P'_{d}:z,\emptyset$ for each  $d\in \overline{X}_D\setminus \widetilde{X}_D$.\footnote{It is important to note that the sets \( A \) and \( B \) must be chosen in such a way that \( Y \) remains an allocation.}

\item[$\boldsymbol{( x^1\to x^2\text{ and }x^3 \to x^2):}$] 
      The application of Condition~\eqref{ecu1 teorema maximal domain} to these pairs of complementary contracts yields the following:

 \begin{equation}\label{ecu5 teorema maximal domain}
 \begin{split}
        x^1\in C^{P'}_h(\overline{X} \cup \{x^1,x^3\}) \text{ and } x^2\notin C^{P'}_h(\overline{X} \cup \{x^2,x^3\}),\text{ and}\\
        x^3\in C^{P'}_h(\overline{X} \cup \{x^1x^3\}) \text{ and } x^2\notin C^{P'}_h(\overline{X} \cup \{x^1,x^2\}).
        \end{split}
    \end{equation}
Now, consider the following preference $P'_h$ satisfying the Condition \eqref{ecu5 teorema maximal domain}:

$$P'_h: \ldots, \overline{X}\cup \{x^1,x^2,x^3\},  \overline{X}\cup \{x^1,x^3\},  \overline{X}\cup \{x^1\}, \overline{X}\cup \{x^3\},  \overline{X},\ldots$$

Define also the following preferences for the rest of the agents:
\begin{align*}
    &P'_{h'}:y^2,y^3,y^1,\emptyset&\\
    &P'_{d^1}:y^1,x^1,\emptyset&\\
    &P'_{d^2}:x^2,y^2,\emptyset&\\
    &P'_{d^3}:y^3,x^3,\emptyset&\\
       &P'_{d}:z,\emptyset&\text{ for each } d\in \overline{X}_D \text{ and } z\in \overline{X}\\
    &P'_{d}:\emptyset&\text{ for each } d\in D \setminus (\overline{X}_D\cup \{d^1,d^2,d^3\})
\end{align*}

Now, we can analyze the possible individually rational allocations and verify that there is no one stable.
\begin{center}
    \begin{tabular}{|c|c|c|}\hline
 $Y_h$    & $y_{h'}$ & Blocking contract\\  \hline 
 $\overline{X}\cup \{x^1,x^2,x^3\}$& $\emptyset$  &  $y^1$ \\

 $\overline{X}\cup \{x^1,x^3\}$& $y^2$  &  $x^2$ \\
  $\overline{X}\cup \{x^1,x^3\}$& $\emptyset$  &  $x^2$ \\

    $\overline{X}\cup \{x^1\}$& $y^2$  &  $x^3$ \\
     $\overline{X}\cup \{x^1\}$& $y^3$  &  $y^2$ \\
         $\overline{X}\cup \{x^1\}$& $\emptyset$  &  $x^2$ \\
       
 $\overline{X}\cup \{x^3\}$& $y^1$  &  $y^2$ \\
  $\overline{X}\cup \{x^3\}$& $y^2$  &  $x^1$ \\
    $\overline{X}\cup \{x^3\}$& $\emptyset$  &  $x^1$ \\
            $\overline{X}$& $B$  &  $y^2$ \\
               
 %$\overline{X}$&  $x^2$ &  $(h,x^1)$ \\
 %$\overline{X}$& $x^1$  & $(h',x^2)$  \\
 $\widetilde{X}\cup A$ & $B$  &  $z$ \\
   \hline
\end{tabular}
\end{center}
where $\widetilde{X}$ is such that $\widetilde{X}\subsetneq\overline{X}$, $A\in2^{\{x^1,x^2,x^3\}}$, $B\in\{\emptyset,y^1,y^2,y^3\}$ and, by the Claim, $z\in C^{P'}_h(\widetilde{X}\cup \{z\})$ together with the fact that $P'_{d}:z,\emptyset$ for each  $d\in \overline{X}_D\setminus\widetilde{X}_D$.\footnote{It is important to note that the sets \( A \) and \( B \) must be chosen in such a way that \( Y \) remains an allocation.}

\item[$\boldsymbol{( x^1\to x^2\text{ and }x^1 \to x^3):}$] 
      The application of Condition~\eqref{ecu1 teorema maximal domain} to these pairs yields the following:

 \begin{equation}\label{ecu6 teorema maximal domain}
 \begin{split}
        x^1\in C^{P'}_h(\overline{X} \cup \{x^1,x^3\}) \text{ and } x^2\notin C^{P'}_h(\overline{X} \cup \{x^2,x^3\}),\text{ and}\\
        x^1\in C^{P'}_h(\overline{X} \cup \{x^1,x^2\}) \text{ and } x^3\notin C^{P'}_h(\overline{X} \cup \{x^2,x^3\}).
        \end{split}
    \end{equation}
Now, consider the following preference $P'_h$ satisfying the Condition \eqref{ecu6 teorema maximal domain}:

$$P'_h: \ldots, \overline{X}\cup \{x^1,x^2,x^3\},  \overline{X}\cup \{x^1,x^2\}, \overline{X}\cup \{x^1,x^3\},  \overline{X}\cup \{x^1\}, \overline{X}\cup \{x^2\},\overline{X}\cup \{x^3\},  \overline{X},\ldots$$

Define also the following preferences for the rest of the agents:
\begin{align*}
    &P'_{h'}:y^2,y^3,y^1,\emptyset&\\
    &P'_{d^1}:y^1,x^1,\emptyset&\\
    &P'_{d^2}:x^2,y^2,\emptyset&\\
    &P'_{d^3}:x^3,y^3,\emptyset&\\
       &P'_{d}:z,\emptyset&\text{ for each } d\in \overline{X}_D \text{ and } z\in \overline{X}\\
    &P'_{d}:\emptyset&\text{ for each } d\in D \setminus (\overline{X}_D\cup \{d^1,d^2,d^3\})
\end{align*}

Now, we can analyze the possible individually rational allocations and verify that there is no one stable.
\begin{center}
    \begin{tabular}{|c|c|c|}\hline
 $Y_h$    & $Y_{h'}$ & Blocking contract\\  \hline 
 $\overline{X}\cup \{x^1,x^2,x^3\}$& $\emptyset$  &  $y^2$ \\

 $\overline{X}\cup \{x^1,x^2\}$& $y^3$  &  $x^3$ \\
  $\overline{X}\cup \{x^1,x^2\}$& $\emptyset$  &  $x^3$ \\

 $\overline{X}\cup \{x^1,x^3\}$& $y^2$  &  $x^2$ \\
  $\overline{X}\cup \{x^1,x^3\}$& $\emptyset$  &  $x^2$ \\
 
    $\overline{X}\cup \{x^1\}$& $y^2$  &  $x^2$ \\
     $\overline{X}\cup \{x^1\}$& $y^3$  &  $x^2$ \\
       $\overline{X}\cup \{x^1\}$& $\emptyset$  &  $x^2$ \\
    
    $\overline{X}\cup \{x^2\}$& $y^1$  &  $y^3$ \\
     $\overline{X}\cup \{x^2\}$& $y^3$  &  $x^1$ \\
       $\overline{X}\cup \{x^2\}$& $\emptyset$  &  $x^1$ \\
    
 $\overline{X}\cup \{x^3\}$& $y^1$  &  $y^2$ \\
  $\overline{X}\cup \{x^3\}$& $y^2$  &  $x^1$ \\
       $\overline{X}\cup \{x^3\}$& $\emptyset$  &  $y^1$ \\
            $\overline{X}$& $B$  &  $y^1$ \\
               
 %$\overline{X}$&  $x^2$ &  $(h,x^1)$ \\
 %$\overline{X}$& $x^1$  & $(h',x^2)$  \\
 $\widetilde{X}\cup A$ & $B$  &  $z$ \\
   \hline
\end{tabular}
\end{center}
where $\widetilde{X}$ is such that $\widetilde{X}\subsetneq\overline{X}$, $A\in2^{\{x^1,x^2,x^3\}}$, $B\in\{\emptyset,y^1,y^2,y^3\}$ and, by the Claim, $z\in C^{P'}_h(\widetilde{X}\cup \{z\})$ together with the fact that $P'_{d}:z,\emptyset$ for each  $d\in \overline{X}_D\setminus\widetilde{X}_D$.\footnote{It is important to note that the sets \( A \) and \( B \) must be chosen in such a way that \( Y \) remains an allocation.}

\item[$\boldsymbol{( x^1\to x^2,~x^2\to x^3\text{ and }x^3 \to x^1):}$] 
      The application of Condition~\eqref{ecu1 teorema maximal domain} to these pairs of complementary contracts yields the following:
  
 \begin{equation}\label{ecu7 teorema maximal domain}
 \begin{split}
        x^1\in C^{P'}_h(\overline{X} \cup \{x^1,x^3\}) \text{ and } x^2\notin C^{P'}_h(\overline{X} \cup \{x^2,x^3\}),\\
        x^2\in C^{P'}_h(\overline{X} \cup \{x^1,x^2\}) \text{ and } x^3\notin C^{P'}_h(\overline{X} \cup \{x^1,x^3\}), \text{ and}\\
        x^3\in C^{P'}_h(\overline{X} \cup \{x^2,x^3\}) \text{ and } x^1\notin C^{P'}_h(\overline{X} \cup \{x^1,x^2\}).
        \end{split}
    \end{equation}
Consider, w.l.o.g., the following preference $P'_h$ that satisfies Condition~\eqref{ecu7 teorema maximal domain}:

$$P'_h: \ldots, \overline{X}\cup \{x^1,x^2,x^3\},  \overline{X}\cup \{x^3\}, \overline{X}\cup \{x^1\}, \overline{X}\cup \{x^2\},  \overline{X},\ldots$$

Recall that $X'=\overline{X}\cup \{x^1,x^2,x^3\}$ is the minimal set of contracts fulfilling Condition \eqref{ecu1 teorema maximal domain}. Observe that $x^1 \notin C^{P'}_h(X' \setminus \{x^2\})$ and $x^2 \notin C^{P'}_h(X' \setminus \{x^1\})$. Then, following the same arguments as in the proof of Lemma~\ref{P_h no pseudo, entonces unicomplementaria}, we can construct a sub-preference of $P'_h$, which contradicts the minimality of $P'_h$. Therefore, this case cannot occur.

   \end{description}  
   
\end{description}
Based on the analysis of the three cases, given a preference relation $P'_h$ that is not pseudo-substitutable, it is possible to construct a preference profile for the remaining agents---all of which are pseudo-substitutable---such that the set of stable allocations is empty, i.e., $S(P) = \emptyset$.
\end{proof}

\section{Appendix: Proof of Propositions \ref{proposicion de propiedad 1} and \ref{proposicion de propiedad 2}}\label{Apendice puebas de proposiciones}

\begin{proof}[Proof of Proposition \ref{proposicion de propiedad 1}]%To simplify the proof, we verify the minimality property for a preference $P_h$. The argument for doctors' preferences is analogous. 
Let $a \in D \cup H$, and let $P'_a$ and $P_a$ be preference relations such that $P'_{a}\sqsubseteq P_{a}$ and  $P'_{a}$ is substitutable.
Assume that $P_a'$  is not minimal. Then, there is $P''_a\sqsubseteq P_{a}$ such that $A(P''_a)\subsetneq A(P'_a)$. Let $X'\subseteq X$ be such that $X'\in A(P'_a)\setminus A(P''_a)$, and let $\overline{X}=C^{P''}_a(X')$. So, $\overline{X}\subsetneq X'$ implying that
$\overline{X}\in A(P''_a)\subsetneq A(P'_a) \subseteq A(P_a).$  
Let $x\in X'\setminus \overline{X}.$ Then, 
\begin{equation}\label{ecu 2 prueba minimal}
    \overline{X}\subseteq \overline{X} \cup \{x\} \subseteq X'.
\end{equation}
Since $P_a'$ is a substitutable preference and $C^{P'}_a(X')=X'$, it follows that $C^{P'}_a(\overline{X} \cup \{x\})=\overline{X} \cup \{x\}.$ Hence, $ \overline{X} \cup \{x\}\in A(P'_a)\subseteq A(P_a)$ which implies that    
\begin{equation}\label{ecu 3 prueba minimal}
    C^{P}_a(\overline{X} \cup \{x\})=\overline{X}\cup \{x\}.
\end{equation}
The fact that $\overline{X}=C^{P''}_a(X')$ and together with  \eqref{ecu 2 prueba minimal} imply that
\begin{equation}\label{ecu 4 prueba minimal}
    C^{P''}_a(\overline{X} \cup \{x\})=\overline{X}.
\end{equation}
Note that equation \eqref{ecu 3 prueba minimal} implies that $x\in C^P_a(\overline{X} \cup \{x\})$ while equation \eqref{ecu 4 prueba minimal} implies that $x\notin C^{P''}_a(\overline{X} \cup \{x\})$. This contradicts the assumption that $P''_a\sqsubseteq P_{a}$. Therefore, $P'_a$ is minimal.\end{proof}

\noindent\begin{proof}[Proof of Proposition \ref{proposicion de propiedad 2}]
  \begin{description}
\item[\textbf{1. $\boldsymbol{A(P'_a)=A(P''_a)}.$}] 
Recall that since $P'_a$ and $P''_a$ are sub-preference relations of $P_a$, we have
\begin{equation}\label{ecu1 propo propiedades de subpreferencias}
    A(P'_a)\subseteq A(P_a)   \text{ and }   A(P''_a)\subseteq A(P_a).
\end{equation}

Assume, for the sake of contradiction, that $A(P'_a)\neq A(P''_a)$. Then, there exists a set of contracts $X_1$ such that
\[
X_1 = C_a^{P'}(X_1)   \text{ and }   X_1 \neq C_a^{P''}(X_1).
\]
Let $X_2 = C_a^{P''}(X_1)$. By definition of the choice function, we have $X_2 \subsetneq X_1$, which, by substitutability, implies that $X_2 \in A(P''_a)$. Pick a contract $x \in X_1 \setminus X_2$. Since $P'$ is substitutable, it follows that $x \in C_a^{P'}(X_2 \cup \{x\})$. 
Observe that $X_2 \subseteq X_2 \cup \{x\} \subseteq X_1$.
 Thus, by the definition of \(X_2\) and the consistency property of the preference \(P''\), we have
\[
C_a^{P''}(X_2 \cup \{x\}) = C_a^{P''}(X_1) = X_2,
\]
which implies that 

\begin{equation}\label{ecu 1 propo aceptables iguales de todas las subpreferencias}
    x \notin C_a^{P''}(X_2\cup\{x\}).
\end{equation}

Since \(X_2 \subsetneq X_1\) and \(x \in X_1 \setminus X_2\), we have
\[
X_2 \cup \{x\} = (X_2 \cup \{x\}) \cap X_1.
\]
Because \(X_1 = C_a^{P'}(X_1)\), it follows that
\[
(X_2 \cup \{x\}) \cap X_1 = (X_2 \cup \{x\}) \cap C_a^{P'}(X_1).
\]
By Remark~\ref{definicion alternatica de contratos subsitutos con inclusiones},
\[
(X_2 \cup \{x\}) \cap C_a^{P'}(X_1) \subseteq C_a^{P'}(X_2 \cup \{x\}).
\]
By the definition of a choice function,
\[
C_a^{P'}(X_2 \cup \{x\}) \subseteq X_2 \cup \{x\}.
\]
Hence,
\[
C_a^{P'}(X_2 \cup \{x\}) = X_2 \cup \{x\},
\]
which implies that \(X_2 \cup \{x\} \in A(P')\). Since \(P' \sqsubseteq P\), we also have \(X_2 \cup \{x\} \in A(P)\), so
\[
X_2 \cup \{x\} = C_a^{P}(X_2 \cup \{x\}).
\]
Because \(P'' \sqsubseteq P\) and \(X_2 \cup \{x\} = C_a^{P}(X_2 \cup \{x\})\), it follows that
\(x \in C_a^{P''}(X_2 \cup \{x\})\), which contradicts \eqref{ecu 1 propo aceptables iguales de todas las subpreferencias}. Therefore, \(A(P'_a) = A(P''_a)\).

\item[\textbf{2. For each $\boldsymbol{X_1, X_2 \in A(P'_a)}$, if $\boldsymbol{X_1 = C_a^{P'}(X_1 \cup X_2)}$, then $\boldsymbol{X_1 = C_a^{P''}(X_1 \cup X_2)}$.}] 

Assume, for the sake of contradiction, that this is not the case, i.e., $X_1 \neq C_a^{P''}(X_1 \cup X_2).$
Let $X_3 = C_a^{P''}(X_1 \cup X_2).$
By the definition of the choice function, 
$X_3 \subseteq X_1 \cup X_2.$
We now analyze two possible cases.

\begin{description}
\item[\textbf{Case 1: $\boldsymbol{X_1\setminus X_3=\emptyset.}$}] 
Since $X_3 \neq X_1$, we have $X_1 \subsetneq X_3$. Hence, for each contract $x \in X_3 \setminus X_1$. 
Note that $C_a^{P'}(X_1) = X_1$, and since $A(P'_a) = A(P''_a)$, it follows that $C_a^{P''}(X_1) = X_1$. 
As $x \notin X_1$, we obtain $x \notin C_a^{P''}(X_1)$. 
By substitutability of $P''$, this further implies $x \notin C_a^{P''}(X_1 \cup X_2) = X_3$, 
contradicting the fact that $x \in X_3 \setminus X_1$. 
Therefore, this case cannot occur.

\item[\textbf{Case 2: $\boldsymbol{X_1\setminus X_3\neq\emptyset.}$}]
For each contract $x \in X_1 \setminus X_3$, we have $
x \in C_a^{P'}(X_3 \cup \{x\}).$
Let $X_4\subseteq X_3$ be such that $X_4 \cup \{x\} = C_a^{P'}(X_3 \cup \{x\}).$ Thus, by definition fo choice function, $X_4 \subseteq X_3$.  
By substitutability of $P'_a$, we get $X_4 \cup \{x\} = C_a^{P'}(X_4 \cup \{x\}),$
implying that $X_4 \cup \{x\} \in A(P'_a)$.  
Since $A(P'_a) = A(P''_a)$, we also have $X_4 \cup \{x\} = C_a^{P''}(X_4 \cup \{x\}).$
We now analyze two subcases:

\textbf{Subcase 2.1: $\boldsymbol{X_4 \subsetneq X_3}$.}  
Then there exists $x' \in X_3 \setminus X_4$.  
We have $x' \notin C_a^{P'}(X_3 \cup \{x\}) = X_4 \cup \{x\}.$
Since  $A(P'_a)=A(P''_a),$ we have that $x' \notin C^{P''}(X_3 \cup \{x\}).$
Moreover, since $X_4 \subsetneq X_3$ and the fact that $P''_a$ is substitutable, we get
\begin{equation}\label{eq:subcase31_not_in_choice}
    x' \notin C_a^{P''}(X_3 \cup \{x'\}).
\end{equation}
Note that $C_a^{P''}(X_1 \cup X_2) = X_3 \subseteq X_3 \cup \{x'\} \subseteq X_1 \cup X_2$.  
By the consistency of $P''_a$, $C_a^{P''}(X_3 \cup \{x\}) = C_a^{P''}(X_1 \cup X_2) = X_3.$
From \eqref{eq:subcase31_not_in_choice}, we get $x' \notin X_3$, contradicting the fact that $x' \in X_3 \setminus X_4$.

\textbf{Subcase 2.2: $\boldsymbol{X_4 = X_3}$.}  
By definition of $X_4$, and the subcase hypothesis, we have that $C_a^{P'}(X_3 \cup \{x\}) = X_3 \cup \{x\}.$
Thus, $X_3 \cup \{x\} \in A(P')$.  
Since $A(P') = A(P'')$, it follows that $X_3 \cup \{x\} = C^{P''}(X_3 \cup \{x\}).$
Hence, $x \in C_a^{P''}(X_3 \cup \{x\}) = C_a^{P''}(C_a^{P''}(X_1 \cup X_2) \cup \{x\}).$
By the path-independence property of $P''_a$, we have $
x \in C_a^{P''}(X_1 \cup X_2 \cup \{x\}).$
Since $x \in X_1$, this implies $x \in C_a^{P''}(X_1 \cup X_2) = X_3,$ contradicting the assumption that $x \in X_1 \setminus X_3$.
\end{description}
Since both subcases lead to a contradiction, Case 2 cannot occur.
\end{description}

Thus, in both cases, we reached a contradiction.  
This contradiction arises from assuming that 
$
X_1 \neq C_a^{P''}(X_1 \cup X_2)
\text{ whenever }X_1 = C_a^{P'}(X_1 \cup X_2).$
Therefore, it must be that 
\[
X_1 = C_a^{P''}(X_1 \cup X_2)
\text{ whenever }X_1 = C_a^{P'}(X_1 \cup X_2).\]
\end{proof}

\end{document}